\documentclass[11pt,letterpaper]{article}

\usepackage[margin=1in]{geometry}
\usepackage{amsfonts}
\usepackage{graphicx}
\usepackage{epstopdf}
\ifpdf
  \DeclareGraphicsExtensions{.eps,.pdf,.png,.jpg}
\else
  \DeclareGraphicsExtensions{.eps}
\fi

\usepackage{nicefrac}
\usepackage{url}
\usepackage{amsmath,amssymb,amsthm}
\usepackage[square,numbers]{natbib}
\usepackage{enumitem}
\setenumerate{wide = 0pt, labelwidth = 1.3333em, labelsep = 0.3333em, leftmargin = \dimexpr\labelwidth + \labelsep\relax, noitemsep,topsep=3pt}
\setitemize{wide = 0pt, labelwidth = 1.3333em, labelsep = 0.3333em, leftmargin = \dimexpr\labelwidth + \labelsep\relax, noitemsep,topsep=3pt}
\usepackage{caption}
\usepackage{xfrac}
\usepackage{hyperref}
\hypersetup{
	colorlinks=true,
	citecolor = blue       
}
\usepackage{booktabs}

\usepackage{algorithm}
\usepackage[ruled,algo2e]{algorithm2e}

\SetAlFnt{\small}
\SetAlCapFnt{\small}
\SetAlCapNameFnt{\small}
\SetAlCapHSkip{0pt}
\IncMargin{-\parindent}
\SetKwInOut{Input}{Input}\SetKwInOut{Output}{Output}

\renewcommand*{\paragraph}[1]{\medskip\noindent{\bfseries #1.\xspace}}

\newtheorem{theorem}{Theorem}[section]
\newtheorem{example}{Example}[section]
\newtheorem{definition}{Definition}[section]
\newtheorem{lemma}{Lemma}[section]
\newtheorem{corollary}{Corollary}[section]
\newtheorem{proposition}{Proposition}[section]
\newtheorem{observation}{Observation}[section]
\newtheorem{claim}{Claim}[section]
\newtheorem{remark}{Remark}[section]

\newcommand*{\myproofname}{My proof}
\newenvironment{myproof}[1][\myproofname]{\begin{proof}[#1]}{\end{proof}}

\newcommand{\final}[1]{\textcolor{black}{#1}}

\title{The Complexity of Contracts\thanks{An extended abstract appeared in \emph{Proceedings of the ACM-SIAM Symposium on Discrete Algorithms (SODA'20)}, Salt Lake City, UT, USA, January 5--8, 2020. Supported by BA/Leverhulme Small Research Grant
SRG1819\textbackslash191601, NSF Award CCF-1813188, ARO grant
W911NF1910294, and the Israel Science Foundation (Grant
No. 336/18). Part of the work of the first author was done while
visiting Google Research. The third author is a Taub Fellow (supported by the Taub Family Foundation).}}
\author{P.~D\"utting\thanks{Department of Mathematics, London School of Economics, Houghton Street, London WC2A 2AE, UK. Email: \tt{
p.d.duetting@lse.ac.uk}.}
\and T.~Roughgarden\thanks{Department of Computer Science, Columbia University,
500 West 120th Street, New York, NY 10027, USA. Email: \tt{tr@cs.columbia.edu}.}
\and I.~Talgam-Cohen\thanks{Department of Computer Science, Technion, Israel Institute of Technology, Haifa, Israel 3200003. Email: \tt{italgam@cs.technion.ac.il}.}}
\date{February 18, 2020}

\begin{document}

\maketitle

\begin{abstract}
  We initiate the study of computing (near-)optimal contracts in succinctly representable principal-agent settings. Here optimality means maximizing the principal's expected payoff over all incentive-compatible contracts---known in economics as ``second-best'' solutions. We also study a natural relaxation to \emph{approximately} incentive-compatible contracts.

We focus on principal-agent settings with succinctly described (and exponentially large) outcome spaces. We show that the computational complexity of computing a near-optimal contract depends fundamentally on the number of agent actions. For settings with a constant number of actions, we present a fully polynomial-time approximation scheme (FPTAS) for the separation oracle of the dual of the problem of minimizing the principal's payment to the agent, and use this subroutine to efficiently compute a $\delta$-incentive-compatible ($\delta$-IC) contract whose expected payoff matches or surpasses that of the optimal IC contract.

With an arbitrary number of actions, we prove that the problem is hard to approximate within any constant $c$. This inapproximability result holds even for $\delta$-IC contracts where $\delta$ is a sufficiently rapidly-decaying function of~$c$. On the positive side, we show that simple linear $\delta$-IC contracts with constant~$\delta$ are sufficient to achieve a constant-factor approximation of the ``first-best'' (full-welfare-extracting) solution, and that such a contract can be computed in polynomial time.
\end{abstract}

\section{Introduction}
\label{sec:introduction}
Economic theory distinguishes three fundamentally different problems
involving asymmetric information and incentives. In the first---known
as \emph{mechanism design} (or \emph{screening})---the less informed
party has to make a decision. A canonical example is Myerson's optimal
auction design problem \cite{Myerson81}, in which a seller wants to
maximize the revenue from selling an item, having only incomplete
information about the buyers' willingness to pay.  The second problem
is known as \emph{signalling} (or \emph{Bayesian persuasion}). Here, as
in the first case, information is hidden, but this time the more
informed party is the active party. A canonical example is Akerlof's
``market for lemons'' \cite{Akerlof70}. In this example, sellers are
better informed about the quality of the products they sell, and may
benefit by sharing (some) of their information with the buyers.

Both of these basic incentive problems have been studied very
successfully and extensively from a computational perspective, see,
e.g.,
\cite{CaiDW12a,CaiDW12b,CaiDW13,BabaioffILW14,CaiDW16,BabaioffGN17,Gonczarowski18,GonczarowskiW18}
and \cite{Dughmi14,DughmiIR14,ChengCDEHT15,DughmiX16}.

The third basic problem, \emph{the agency problem} in \emph{contract
theory}, has received far less attention from the theoretical computer
science community, despite being regarded as equally important in
economic theory (see, e.g., the scientific background on the 2016
Nobel Prize for Hart and Holmstr\"om \cite{Nobel16}).  (A notable
exception is \cite{Babaioff-et-al12}, which we will discuss with
further related work in more detail below.)

The basic scenario of contract theory is captured by the following
\emph{hidden-action principal-agent problem} \cite{GrossmanHart83}:
There is one \emph{principal} and one \emph{agent}. The agent can take
one of $n$ actions $a_i \in A_n$. Each action~$a_i$ is associated with a
distribution $F_{i}$ over $m$ outcomes $x_j \in \mathbb{R}_{\geq 0}$,
and has a cost $c_i \in \mathbb{R}_{\geq 0}$. The principal designs a
contract $p$ that specifies a payment $p(x_j)$ for each outcome
$x_j$. The agent chooses an action $a_i$ that maximizes expected
payment minus cost, i.e., $\sum_j F_{i,j} p(x_j) - c_i$. The principal
seeks to set up the contract so that the chosen action maximizes
expected outcome minus expected payment, i.e., $\sum_j F_{i,j} x_j -
\sum_j F_{i,j} p(x_j)$.

The principal-agent problem is quite different from mechanism design
and signalling, where the basic difficulty is the information
asymmetry and that part of the information is hidden. In the
principal-agent problem the issue is one of \emph{moral hazard}: in
and by itself the agent has no intrinsic interest in the expected
outcome to the principal.

It is straightforward to see that the optimal contract can be found in
time polynomial in $n$ and $m$ by solving $n$ linear programs (LPs). For
each action the corresponding LP gives the smallest
expected payment at which this action can be implemented. The action
that yields the highest expected reward minus payment gives the optimal payoff to
the principal, and the LP for this action the optimal contract.

\paragraph{Succinct principal-agent problems}
This linear programming-based algorithm for computing an optimal
contract has several analogs in algorithmic game theory:
\begin{enumerate}

\item {\em Mechanism design.}  For many basic mechanism design
problems, the optimal (randomized) mechanism is the solution of a linear
program with size polynomial in that of the players' joint type space.

\item {\em Signalling.} For many computational problems in signalling,
the optimal signalling scheme is the solution to a linear program with
size polynomial in the number of receiver actions and possible states
of nature.

\item {\em Correlated equilibria.} In finite games, a correlated
equilibrium can be computed using a linear program with size
polynomial in the number of game outcomes.

\end{enumerate}
These linear-programming-based solutions are unsatisfactory when their
size is exponential in some parameter of interest.  For example, in
the mechanism design and correlated equilibria examples, the size of
the LP is exponential in the number of players.  A major contribution
of theoretical computer science to game theory and economics has been
the articulation of natural classes of succinctly representable
settings and a thorough study of the computational complexity of
optimal design problems in such settings.  Examples include work on 
multi-dimensional mechanism design that has emphasized succinct type
distributions \cite{CaiDW12a,CaiDW12b,CaiDW13,CaiDW16}, succinct
signalling schemes with an exponential number of states of nature
\cite{DughmiX16}, and the efficient computation of correlated equilibria
in succinctly
representable multi-player games \cite{PapadimitriouR08,JiangL15}.
The goal of this paper is to initiate an analogous line of work for 
succinctly described agency problems in contract theory.

We focus on principal-agent settings with succinctly described (and
exponentially large) outcome spaces,
along with a reward function that supports value queries and a
distribution for each action with polynomial description.
While there are many such settings one can study, we focus on what is
arguably the most natural one from a theoretical computer science
perspective, where outcomes correspond to vertices of the hypercube,
the reward function is additive, and the distributions are product
distributions.  (Cf., work on computing revenue-maximizing multi-item
auctions with product distributions over additive valutions, e.g.\
\cite{CaiDW12a,CaiDW12b}.)

For example, outcomes could correspond to sets of items, where items
are sold separately using posted prices.  Actions could correspond to
different marketing strategies (with different costs), which lead to
different (independent) probabilities of sales of various items.  Or,
imagine that a firm (principal) uses a headhunter (agent) to hire an
employee (action).  Dimensions could correspond to tasks or skills.
Actions correspond to types of employees, costs correspond to the
difficulty of recruiting an employee of a given type, and for each
employee type there is some likelihood that they will possess each
skill (or be able to complete some task).  
\final{The firm wants to motivate the headhunter to put in enough effort to recruit an employee who is likely to have useful skills for the firm, without actually running extensive interviews to find out the employee's type.}

In our model, as in the classic model, there is a \emph{principal} and
an \emph{agent}. The agent can take one of $n$ actions $a_i \in A_n$,
and each action has a cost $c_i \in \mathbb{R}_{\geq 0}$. Unlike in
the original model, we are given a set of items $M$, with $|M| =
m$. Outcomes correspond to subsets of items $S \in 2^M$. 
Each item has a reward $r_j$, and the reward of a set $S$ of items is $\sum_{j \in S} r_j$.
Every action
$a_i$ comes with probabilities $q_{i,j}$ for each item $j$. If action
$a_i$ is chosen, each item $j$ is included in the outcome independently
with probability $q_{i,j}$. A contract specifies a payment $p_S$ for
each outcome $S \in 2^M$. 
The goal is to compute a contract that maximizes (perhaps approximately)
the principal's payoff 
\final{in running time
polynomial in $n$ and $m$ (which is logarithmic in the size $|2^M|$ of
the outcome space).}
\final{Note that if we wish to describe the contract by its nonzero payments, the running time requirement 
forces us to use only contracts with polynomial description length, that is, zero payments for all but polynomially-many $S\in 2^M$.}

\paragraph{A notion of approximate IC for contracts}
The classic approach in contract theory is to require that the agent
is incentivized exactly, i.e., he (weakly) prefers the chosen action
over every other action.  We refer to such contracts as incentive
compatible or just IC contracts.  Motivated in part by our hardness
results for IC contracts
(see the next section) and inspired by the success of notions of
approximate incentive compatibility in mechanism design (see
\cite{Cai13,Weinberg14,CaiDW16}, hereafter referred to as the
\emph{CDW framework}), we introduce a notion of approximate incentive
compatibility that is suitable for contracts.

Our notion of $\delta$-incentive compatibility
(or $\delta$-IC) is that the agent utility of the approximately
incentivized action
$a_i$ is at least that of any other action $a_{i'}$, less~$\delta$.
(See
Section~\ref{sec:prelim} and Appendix~\ref{appx:prelim} for details
and discussion.)  This notion is natural for several reasons.  First,
it coincides with the usual notion of $\epsilon$-IC
in ``normalized'' mechanism design settings (with all valuations
between~0 and~1), as in \cite{Cai13,Weinberg14}.
A second reason is behaviorial.  There is an increasing body of work
in economics on 
behavioral biases in contract theory \cite{Koszegi14}, including
strong empirical evidence that such biases play an important role in
practice---for example, that agents ``gift'' effort to the principals
employing them \cite{Akerlof82}. 
The notion of $\delta$-IC offers a mathematical formulation of an agent's bias.
Along similar lines, 
\cite{Carroll13} advocates generally for approximate IC
constraints in settings where the designer can propose their
``preferred action'' to agents, in which case an agent may be biased
against deviating
due to the complexities involved in determining the agent-optimal action or the
psychological costs of deviating.  See also \cite{EnglmaierL12} for
related discussion in the context of contract theory.

\subsection{Our contribution and techniques}

We prove several positive and negative algorithmic results for
computing near-optimal contracts in succinctly described
principal-agent settings.  Our work reveals a fundamental dichotomy
between settings with a constant number of actions and those with an
arbitrary number of actions.

\paragraph{Constant number of actions}
For a constant number of actions, we prove in 
Section~\ref{sec:constant} that while it is $NP$-hard to compute an
optimal IC contract, there is an FPTAS that computes a $\delta$-IC contract with expected principal surplus at least that of the optimal IC contract; the running time is polynomial in~$m$ and~$1/\delta$.

\begin{theorem}[See Theorem~\ref{thm:FPTAS}, Corollary~\ref{cor:FPTAS}]
For every constant $n \geq 1$ and $\delta > 0$, there is an algorithm
that computes a $\delta$-IC contract with expected principal surplus
at least that of an optimal IC contract in time polynomial in $m$ and $1/\delta$.
\end{theorem}

The starting point of our algorithm is a linear programming
formulation of the problem of incentivizing a given action with the
lowest possible expected payment.  Our formulation has a polynomial
number of constraints (one per action other than the
to-be-incentivized one) but an exponential number of variables (one
per outcome).  A natural idea is to then solve the dual linear program
using the ellipsiod method.  The dual separation oracle is: given a weighted
mixture of $n-1$ product distributions (over the $m$ items)
and a reference
product distribution~$q^*$, minimize the ratio of the probability of
outcome $S$ in the mixture distribution and that in the reference
distribution.
Unfortunately, as we show, this is an 
$NP$-hard problem, even when there are only~$n=3$ actions.
On the other hand, we provide an FPTAS for the separation oracle in
the case of a constant number of actions, based on a delicate
multi-dimensional bucketing approach.  The standard method of
translating an FPTAS for a separation oracle to an FPTAS for the
corresponding linear program relies on a scale-invariance property
that is absent in our problem.  We proceed instead via a strengthened
version of our dual linear program, to which our FPTAS separation
oracle still applies, and show how to extract from an approximately
optimal dual solution a $\delta$-IC contract with objective function
value at least that of the optimal solution to the original linear
program.

\paragraph{Arbitrary number of actions}
The restriction to a constant number of actions is essential for the
positive results above (assuming $P \neq NP$).
Specifically, we prove in Section~\ref{sec:hardness} that computing the IC
contract that maximizes the expected payoff to the principal is
$NP$-hard, even to approximate to within any constant
$c$. This hardness of approximation result persists even if we relax
from exact IC to $\delta$-IC contracts, provided $\delta$ is
sufficiently small as a function of $c$.

\begin{theorem}[See Theorem~\ref{thm:hardness-of-approx}, Corollary~\ref{cor:hardness-of-approx}]
For every constant $c \in \mathbb{R}$, $c\geq 1$, it is $NP$-hard to
find a IC contract that approximates the optimal expected payoff
achievable by an IC contract to within a multiplicative factor of $c$.
\end{theorem}

\begin{theorem}[See Theorem~\ref{thm:hardness-of-approx}, Corollary~\ref{cor:hardness-of-approx-delta}]
For any constant $c \in \mathbb{R}, c \geq 5$ and $\delta \leq
(\frac{1}{4c})^c$, it is $NP$-hard to find a $\delta$-IC contract that
guarantees $> \frac{2}{c} \text{OPT}$, where $\text{OPT}$ is the
optimal expected payoff achievable by an IC contract.
\end{theorem}

We prove these hardness of approximation results by reduction from
MAX-3SAT, using the fact that it is $NP$-hard to distinguish between a
satisfiable MAX-3SAT instance and one in which there is no assignment
satisfying more than a $7/8+\alpha$ fraction of the clauses, where
$\alpha$ is some arbitrarily small constant \cite{Hastad01}. Our
reduction utilizes the gap between ``first best''
(full-welfare-extracting) and ``second best'' solutions in contract
design settings, where satisfiable instances of MAX-3SAT map to
instances where there is no gap between first and second best and
instances of MAX-3SAT in which no more than $7/8+\alpha$ clauses can
be satisfied map to instances where there is a constant-factor
multiplicative gap between the first-best and second-best solutions.

On the positive side, we prove
that for every constant $\delta$ there
is a simple (in fact, linear\footnote{A linear contract is defined by
a single parameter $\alpha \in [0,1]$, and sets the payment $p_S$ for
any set $S \in 2^M$ to 
$p_S = \alpha \cdot \sum_{j \in S} r_j$. 
Linear
contracts correspond to a simple percentage commission, and are
arguably among the most frequently used contracts in practice. See
\cite{Carroll15} and \cite{DRT18} for recent work in economics and
computer science in support of linear contracts.}) contract that
achieves a $c_\delta$-approximation, where $c_\delta$ is a constant
that depends on $\delta$. This approximation guarantee is with
respect to the strongest possible benchmark, the first-best
solution.\footnote{Note that the principal's objective function
(reward minus payment to the agent) is a mixed-sign objective; such
functions are generally challenging for relative approximation
results.}

\begin{theorem}[See Theorem~\ref{thm:delta-ic-approx}]
For every constant $\delta > 0$ there exists a constant $c_\delta$ and a
polynomial-time (in $n$ and $m$) computable $\delta$-IC contract that
obtains a multiplicative $c_\delta$-approximation to the optimal
welfare.
\end{theorem}

Our proof of this result, in Section~\ref{sec:approx}, shows that the
optimal social welfare can be upper bounded by a sum of (constantly
many in $\delta$) expected payoffs achievable by $\delta$-IC
contracts. The best such contract thus obtains a constant
approximation to the optimal welfare.

\paragraph{Black-box distributions} 
Product distributions are a rich and natural class of succinctly
representable distributions to study, but one could also consider
other classes.  Perhaps the strongest-imaginable positive result would
be an efficient algorithm for computing a near-optimal contract that
works with {\em no} assumptions about each action's probability
distribution over outcomes, other than the ability to sample from them
efficiently.  (Positive examples of this sort in signalling
problem include~\cite{DughmiX16} and in mechanism design include~\cite{HL}
and its many follow-ups.)  Interestingly, the principal-agent problem
poses unique challenges to such ``black-box'' positive results.  The
moral reason for this is explained, for example, in
\cite{Salanie05}: Rewards play a dual role in contract settings,
both defining the surplus from the joint project to be shared between
the principal and agent {\em and} providing a signal to the principal
of the agent's action.  For this reason, in optimal contracts, the
payment to the agent in a given outcome is governed both by the
outcome's reward and on its ``informativeness,'' and the latter is
highly sensitive to the precise probabilities in the outcome
distributions associated with each action.  In
Section~\ref{sec:black-box} we translate this intuition into an
information-theoretic impossibility result for the black-box model,
showing that positive results are possible only under strong
assumptions on the distributions (e.g., that the minimum non-zero
probability is bounded away from~0).

\subsection{Related work}

The study of computational aspects of contract theory was pioneered by
Babaioff, Feldman and Nisan \cite{Babaioff-et-al12} (see also their
subsequent works, notably \cite{EmekF12} and
\cite{BabaioffWinter14}). This line of work studies a problem
referred to as \emph{combinatorial agency}, in which combinations of
agents replace the single agent in the classic principal-agent
model. The challenge in the new model stems from the need to
incentivize multiple agents, while the action structure of each agent
is kept simple (effort/no effort). The focus of this line of work is
on complex combinations of agents' efforts influencing the outcomes,
and how these determine the subsets of agents to contract with. The
resulting computational problems are very different from the
computational problems in our model.\footnote{For example, several of
the key computational questions in their problem turn out to be
$\#P$-hard, while all of the problems we consider are in $NP$.}

A second direction of highly related work is \cite{MA18}. This work considers a principal-agent model in which the agent action space is exponentially sized but compactly represented, and argue that in such settings indirect (interactive) mechanisms can be better than one-shot mechanisms. 
Our focus is more algorithmic, and instead of a compactly represented action space we consider a compactly represented outcome space.

A third direction of related work considers a bandit-style model for contract design \cite{HoSV16}. In their model each arm corresponds to a contract, and they present a procedure that starts out with a discretization of the contract space, which is adaptively refined, and which achieves sublinear regret in the time horizon. Again the result is quite different from our work, where the complexity comes from the compactly represented outcome space, and our result on the black-box model sheds a more negative light on the learning approach.

Further related work comes from Kleinberg and Kleinberg~\cite{KleinbergK18} who consider the problem of delegating a task to an agent in a setting where (unlike in our model) monetary compensation is not an option. Although payments are not available, they show through an elegant reduction to the prophet-inequality problem that constant competitive solutions are possible.

A final related line of work was initiated by Carroll~\cite{Carroll15} who---working in the classic model (where computational complexity is not an issue)---shows a sense in which linear contracts are max-min optimal (see also the recent work of \cite{WaltonC19}). D\"utting et al.~\cite{DRT18} show an alternative such sense, and also provide tight approximation guarantees for linear contracts.

\section{Preliminaries}
\label{sec:prelim}
We start by defining succinct principal-agent settings and the contract design problem.

\subsection{Succinct principal-agent settings} 

Let $n$ and $m$ be parameters.
A principal-agent setting is composed of the following: $n$~\emph{actions} $A_n$ among which the agent can choose, and their \emph{costs} $0=c_1\le \dots \le c_n$ for the agent; \emph{outcomes} which the actions can lead to, and their \emph{rewards} for the principal; and a \emph{mapping} from actions to distributions over outcomes. 
Crucially, the agent's choice of action is hidden from the principal, who observes only the action's realized outcome. 
Our goal is to study \emph{succinct} principal-agent settings with description size polynomial in $n$ and $m$; the (implicit) outcome space can have size exponential in~$m$. Throughout, unless stated otherwise, all principal-agent settings we consider are succinct. 
We focus on arguably one of the most natural models of succinctly-described settings, namely those with \emph{additive} rewards and \emph{product} distributions. 

In more detail, let $M = \{1,2,..,m\}$, where $M$ is referred to as the \emph{item set}. 
Let the outcome space be $\{0,1\}^M$, that is, every outcome is an item subset $S \subseteq M$. 
For every item $j\in M$, the principal gets an additive reward $r_j$ if the realized outcome includes $j$, so the principal's reward for outcome $S$ is $r_S=\sum_{j \in S} r_j$. 
Every action $a_i\in A_n$ is associated with probabilities $q_{i,1},...,q_{i,m} \in [0,1]$ for the items. We denote the corresponding product distribution by $q_i$.
When the agent takes action $a_i$, item $j$ is included in the realized outcome independently with probability $q_{i,j}$. The probability of outcome $S$ is thus
$q_{i,S}=(\prod_{j\in S}q_{i,j})(\prod_{j\notin S}(1-q_{i,j}))$.
By linearity of expectation, the principal's expected reward given action $a_i$ is $R_i=\sum_{S} q_{i,S} r_S=\sum_{j} q_{i,j}r_j$. Action $a_i$'s expected \emph{welfare} is $R_i-c_i$, and we assume $R_i-c_i\ge 0$ for every $i\in [n]$.

\begin{example}[Succinct principal-agent setting]
A company (principal) hires an agent to sell its $m$ products. The agent may succeed in selling any subset of the $m$ items, depending on his effort level, where the $i$th level leads to sale of item $j$ with probability $q_{i,j}$. Reward $r_j$ from selling item $j$ is the profit-margin of product $j$ for the company.
\end{example}

\paragraph{Representation} A succinct principal-agent setting is described by an $n$-vector of costs $c$, an $m$-vector of rewards $r$, and an $n\times m$-matrix $Q$ where entry $(i,j)$ is equal to probability $q_{i,j}$ (and we assume for simplicity that the number of bits of precision for all values is poly$(n,m)$).

\paragraph{Assumptions}
Unless stated otherwise, we assume that all principal-agent settings are \emph{normalized}, i.e., $R_i\le 1$ for every $a_i\in A_n$ (and thus also $c_i\le 1$). Normalization amounts to a simple change of ``currency'', i.e., of the units in which rewards and costs are measured. 
It is a standard assumption in the context of approximate incentive compatibility---see Section~\ref{sub:relax} (similar assumptions appear in both the CDW framework and in \cite{Carroll13}).
We also assume \emph{no dominated actions}: every two actions $a_i,a_{i'}$ have distinct expected rewards $R_{i}\ne R_{i'}$, and $R_{i'}<R_{i}$ implies $c_{i'}<c_i$. That is, we assume away any action that simultaneously costs more for the agent and has lower expected reward for the principal than some (dominating) alternative action. 

\paragraph{Contracts and incentives}
A \emph{contract} $p$ is a vector of payments from the principal to the agent. Payments are non-negative; this is known as \emph{limited liability} of the agent.%
\footnote{Limited liability plays a similar role in the contract literature as risk-averseness of the agent. Both reflect the typical situation in which the principal has ``deeper pockets'' than the agent and is thus the better bearer of expenses/risks.}
The contractual payments are contingent on the outcomes and not actions, as the actions are not directly observable by the principal.
A contract $p$ can potentially specify a payment $p_S\ge 0$ for every outcome~$S$, but by linear programming (LP) considerations detailed below, we can focus on contracts \final{for which the support size of the vector $p$ is polynomial in $n$.}
We sometimes denote by $p_i$ the expected payment $\sum_{S\subseteq M} q_{i,S}p_S$ to the agent for choosing action~$a_i$,
and without loss of generality restrict attention to contracts for which $p_i\le R_i$ for every $a_i\in A_n$ (in particular, $p_i\le 1$ by normalization).

Given contract $p$, 
the agent's expected \emph{utility} from choosing action $a_i$ is $p_i-c_i$. 
The principal's expected \emph{payoff} is then $R_i-p_i$.
The agent wishes to maximize his expected utility over all actions and over an outside option with utility normalized to zero (``individual rationality'' or \emph{IR}). Since by assumption the cost $c_1$ of action $a_1$ is $0$, the outside opportunity is always dominated by action $a_1$ and so we can omit the outside option from consideration. Therefore, the incentive constraints for the agent to choose action $a_i$ are: $p_i-c_i\ge p_{i'}-c_{i'}$ for every $i'\ne i$. If these constraints hold we say $a_i$ is \emph{incentive compatible (IC)} (and as discussed, in our model IC implies IR). 
The standard tie-breaking assumption in the contract design literature is that among several IC actions the agent tie-breaks in favor of the principal, i.e.~chooses the IC action that maximizes the principal's expected payoff.%
\footnote{The idea is that one could perturb the payment schedule slightly to make the desired action uniquely optimal for the agent. For further discussion see \cite[p. 8]{CaillaudHermalin00}.} 
We say contract $p$ \emph{implements} or \emph{incentivizes} action $a_i$ if given $p$ the agent chooses~$a_i$ (namely $a_i$ is IC and survives tie-breaking). If there exists such a contract for action $a_i$ we say $a_i$ is \emph{implementable}, and slightly abusing notation we sometimes refer to the implementing contract as an \emph{IC contract}.

\paragraph{Simple contracts} 
In a \emph{linear} contract, the payment scheme is a linear function of the rewards, i.e., $p_S = \alpha r_S$ for every outcome $S$. We refer to $\alpha\in[0,1]$ as the linear contract's \emph{parameter}, and it serves as a succinct representation of the contract. Linear contracts have an alternative succinct representation by an $m$-vector of item payments $p_j=\alpha r_j$ for every $j\in M$, which induce additive payments $p_S=\sum_{j\in S}p_j$.
A natural generalization is \emph{separable} contracts, the payments of which can also be separated over the $m$ items and represented by an $m$-vector of non-negative payments (not necessarily linear). 
The optimal linear (resp., separable) contract can be found in polynomial time (see Proposition~\ref{pro:separable-poly-time} in Appendix~\ref{appx-sub:delta-IC-properties}). We return to linear contracts in Section~\ref{sec:approx} and to separable contracts in Appendix~\ref{appx:separable}.

\subsection{Contract design and relaxations}
\label{sub:relax}

The goal of contract design is to maximize the principal's expected payoff from the action chosen by the agent subject to IC constraints.
A corresponding computational problem is OPT-CONTRACT: The input is a succinct principal-agent setting, and the output is the principal's expected payoff from the optimal IC contract. 
A related problem is MIN-PAYMENT: The input is a succinct principal-agent setting and an action $a_i$, and the output is the minimum expected payment $p^*_i$ with which $a_i$ can be implemented (up to tie-breaking).
OPT-CONTRACT reduces to solving $n$ instances of MIN-PAYMENT to find $p^*_i$ for every action $a_i$, and returning the maximum expected payoff to the principal $\max_{i\in[n]}\{R_i-p^*_i\}$. 
Observe that MIN-PAYMENT can be formulated as an exponentially-sized LP \final{with $2^m$ variables $\{p_S\}$ (one for each set $S \subseteq M$) and $n-1$ constraints:} 
\begin{align}
\min~& \sum_{S\subseteq M} {q_{i,S} p_S} \label{LP:min-pay}\\
\text{s.t.}~& \sum_{S\subseteq M} {q_{i,S} p_S} - c_i \ge \sum_{S\subseteq M} {q_{i',S} p_S} - c_{i'} && \forall i'\ne i,i'\in[n], \notag\\
& p_S \ge 0 && \forall S\subseteq M.\notag
\end{align}
%

\final{The dual LP has  $n-1$ nonnegative variables $\{\lambda_{i'}\}$ (one for every action $i'$ other than $i$), and exponentially-many constraints:}
\begin{align}
\max~& \sum_{i'\ne i} {\lambda_{i'}(c_i-c_{i'})} \label{LP:dual} \displaybreak[0]\\
\text{s.t.}~& \big(\sum_{i'\ne i} {\lambda_{i'}}\big) - 1 \le \sum_{i'\ne i} {\lambda_{i'}\frac{q_{i',S}}{q_{i,S}}} &&\forall S\subseteq E,q_{i,S}>0, \notag\displaybreak[0]\\
& \lambda_{i'} \ge 0 &&\forall i'\ne i,i'\in [n].\notag
\end{align}


Standard duality considerations imply that there always exists a succinct optimal contract with $n-1$ nonzero payments.
However, the ellipsoid method cannot be applied to solve the dual LP in polynomial time. The separation oracle, which is related to the concept of likelihood ratios from statistical inference, turns out to be NP-hard except for the $n=2$ case---see Proposition \ref{pro:separation-NP-hard} in Appendix~\ref{appx:ellipsoid}.


We return to LP~\eqref{LP:min-pay} and to its dual LP~\eqref{LP:dual} in Section~\ref{sec:constant}.
 
\paragraph{Relaxed IC}
Contract design like auction design is ultimately an optimization problem subject to IC constraints. The state-of-the-art in optimal \emph{auction} design requires a relaxation of IC constraints to $\epsilon$-IC. In the CDW framework, the $\epsilon$ loss factor is additive and applies to normalized auction settings.
The framework enables polytime computation of an $\epsilon$-IC auction with expected revenue approximating that of the optimal IC auction.%
\footnote{To be precise, the CDW framework focuses on \emph{Bayesian} IC (BIC) and $\epsilon$-BIC auctions.}   
Appropriate $\epsilon$-IC relaxations are also studied in multiple additional contexts---see \cite{Carroll13} and references within for voting, matching and competitive equilibrium; and \cite{Papadimitriou06} for Nash equilibrium. 
We wish to achieve similar results in the context of optimal contracts. 
For completeness we include the definition of $\epsilon$-IC cast in the language of contracts: 

\begin{definition}[$\delta$-IC action]
	Consider a (normalized) contract setting.
	For $\delta\ge0$, an action $a_i$ is \emph{$\delta$-IC} given a contract $p$ if
	the agent loses no more than additive $\delta$ in expected utility by choosing~$a_i$, i.e.:
	$p_i-c_i \ge p_{i'}-c_{i'}-\delta$ for every action $a_{i'}\ne a_i$.
\end{definition}

(Notation-wise, we will sometimes replace $\delta$ by $\Delta$ and refer to $\Delta$-IC actions.)
As in the IC case, we often slightly abuse notation and refer to the contract $p$ itself as $\delta$-IC. By this we mean a contract $p$ with an (implicit) action $a_i$ that is $\delta$-IC given $p$ (if there are several such $\delta$-IC actions, by our tie-breaking assumption the agent chooses the one that maximizes the principal's expected payoff).
We also say the contract \emph{$\delta$-implements} or \emph{$\delta$-incentivizes} action $a_i$. Finally if there exists such a contract for $a_i$ then we say this action is \emph{$\delta$-implementable}. Interestingly, by LP duality, any action can be $\delta$-implemented up to tie-breaking even for arbitrarily small $\delta$ (Proposition~\ref{pro:delta-implementable} in Appendix~\ref{appx:prelim}).
We denote by $\delta$-OPT-CONTRACT and $\delta$-MIN-PAYMENT the above computational problems with IC replaced by $\delta$-IC (e.g., the input to $\delta$-OPT-CONTRACT is a succinct principal-agent setting and a parameter $\delta$, and the output is the principal's expected payoff from the optimal $\delta$-IC contract).

In Appendix~\ref{appx-sub:IC-vs-delta} we study the relation between optimal IC and $\delta$-IC contracts. We show that for every $\delta$-IC contract there is an IC contract with approximately the same expected payoff to the principal up to small---and necessary---multiplicative and additive losses. Thus relaxing IC to $\delta$-IC increases the expected payoff of the principal only to a certain extent. More precisely, Proposition~\ref{pro:from-delta-to-IC} shows that any $\delta$-IC contract can be easily transformed into an IC contract that maintains at least $(1-\sqrt{\delta})$ of the principal's expected payoff up to an additive loss of $(\sqrt{\delta}-\delta)$.
Similar results are known in the context of \emph{auctions} (see \cite{HartlineKM15,DughmiHKN17} for welfare maximization and \cite{DaskalakisW12} for revenue maximization).%
\footnote{We thanks an anonymous reviewer for pointing us to these references.}
Proposition \ref{pro:delta-IC-much-better} shows that an additive loss is necessary, as even for tiny $\delta$ there can be a multiplicative constant-factor gap between the expected payoff of an IC contract and a $\delta$-IC one.

\paragraph{Relaxed IC with exact IR}
In our model, IC implies IR due to the existence of a zero-cost action $a_1$, but this is no longer the case for $\delta$-IC.
What if we are willing to relax IC to $\delta$-IC due to the considerations above, but do not want to give up on IR? 
\final{Suppose} we enforce IR by assuming that the agent chooses a $\delta$-IC action only if it has expected utility $\ge 0$. The following lemma shows that this has only a small additive effect on the principal's expected payoff, allowing us from now on to focus on $\delta$-IC contracts (IR can be later enforced by applying the lemma):

\begin{lemma}
	\label{lem:from-delta-to-IR} 
	For every $\delta$-IC contract $p$ that achieves expected payoff of $\Pi$ for the principal, there exists a $\delta$-IC \emph{and IR} contract $p'$ that achieves expected payoff of $\ge \Pi-\delta$.
\end{lemma}

\begin{proof}
	Fix a principal-agent setting. Let $a_i$ be the action $\delta$-incentivized by contract $p$ and assume $a_i$ is not IR. Observe that the agent's expected utility from $a_i$ is $\ge -\delta$ (otherwise $a_i$ would not be $\delta$-IC with respect to $a_1$, which has expected utility $\ge 0$ for the agent).
	First, if $\Pi > \delta$, then let $p'$ be identical to $p$ except for an additional $\delta$ payment for every outcome. Contract $p'$ still $\delta$-incentivizes action $a_i$, but now the agent's expected utility from $a_i$ is $\ge 0$, as required.
	Otherwise if $\Pi\le \delta$, let $p'$ be the contract with all-zero payments. The expected payoff to the principal is zero, which is at most an additive $\delta$ loss compared to $\Pi$. 
\end{proof}

\section{Constant number of actions}
\label{sec:constant}

In this section we begin our exploration of the computational problems OPT-CONTRACT and MIN-PAYMENT by considering principal-agent settings with a constant number $n$ of actions.
For every constant $n\ge 3$ these problems are NP-hard, and this holds even if the IC requirement is relaxed to $\delta$-IC 
(See Proposition~\ref{pro:implement-cost-hard} and Corollary~\ref{cor:implement-cost-hard} in Appendix~\ref{appx:constant}). 
As our main positive result, we establish the tractability of finding a $\delta$-IC contract that matches the expected payoff of the optimal IC contract.
In Section~\ref{sec:hardness} we show this result is too strong to hold for non-constant values of $n$ (under standard complexity assumptions), and in Section~\ref{sec:approx} \final{we} provide an approximation result for general settings.

To state our results more formally, fix a principal-agent setting and action $a_i$; let $OPT_i$ be the solution to MIN-PAYMENT for~$a_i$ (or $\infty$ if $a_i$ cannot be implemented up to tie-breaking without loss to the principal); and let $OPT$ be the solution to OPT-CONTRACT. Observe that $OPT=\max_{i\in[n]}\{R_i-OPT_i\}$. Our main results in this section are the following:

\begin{theorem}[MIN-PAYMENT]
	\label{thm:FPTAS}
	There exists an algorithm that receives as input a (succinct) principal-agent setting with a constant number of actions and $m$ items, an action $a_i$, and a parameter $\delta > 0$, and returns in time poly$(m,\frac{1}{\delta})$ a contract that $\delta$-incentivizes $a_i$ with expected payment $\le OPT_i$ 
	to the agent.  
\end{theorem}

\begin{corollary}[OPT-CONTRACT]
	\label{cor:FPTAS}
	There exists an algorithm that receives as input a (succinct) principal-agent setting with a constant number of actions and $m$ items, and a parameter $\delta > 0$, and returns in time poly$(m,\frac{1}{\delta})$ a $\delta$-IC contract with expected payoff $\ge OPT$ to the principal. 
\end{corollary}

\begin{proof}
	Apply the algorithm from Theorem \ref{thm:FPTAS} once per action $a_i$ to get a contract that $\delta$-incentivizes $a_i$ with expected payoff at least 
	$R_i-OPT_i$ to the principal. Maximizing over the actions we get a $\delta$-IC contract with expected payoff $\ge OPT$ to the principal. 
\end{proof}

Corollary \ref{cor:FPTAS} shows how to achieve $OPT$ with a $\delta$-IC contract rather than an IC one, in the same vein as the CDW results for auctions. A similar result does not hold for general $n$ unless P$=$NP (Corollary \ref{cor:hardness-of-approx-delta}). Note that the $\delta$-IC contract can be transformed into an IR one with an additive~$\delta$ loss by applying Lemma~\ref{lem:from-delta-to-IR}, and to a fully IC one with slightly more loss by Proposition~\ref{pro:from-delta-to-IC}, where $\delta$ can be an arbitrarily small inverse polynomial in $m$. 

In the rest of the section we prove Theorem \ref{thm:FPTAS}. 


\paragraph{An FPTAS for the separation oracle}
We begin by stating the separation oracle problem.
LP~\eqref{LP:min-pay} formulates MIN-PAYMENT for action~$a_i$. Its dual LP~\eqref{LP:dual} has constraints of the form: 
\begin{equation}
(\sum_{i'\ne i} {\lambda_{i'}}) - 1 \le \sum_{i'\ne i} {\lambda_{i'}\frac{q_{i',S}}{q_{i,S}}}.\label{eq:dual-constraint}
\end{equation}
We can rewrite \eqref{eq:dual-constraint} as 
$1 - 1/(\sum_{i' \neq i} \lambda_{i'}) \le \frac{1}{q_{i,S}} \cdot \sum_{i' \neq i} (\lambda_{i'} / (\sum_{i' \neq i} \lambda_{i'})) \cdot q_{i',S}$. 
Thus the separation oracle problem for dual LP~\eqref{LP:dual} is in fact the following problem: Let $n$ be a constant and $m$ a parameter. The input is $n-1$ nonnegative weights $\{\alpha_{i'}\}$ 
that sum to~1; $n-1$ product distributions $\{q_{i'}\}$; 
and a product distribution~$q_i$; where all product distributions are over $m$ items~$M$. The goal is to minimize the likelihood ratio $\frac{\sum_{i'} {\alpha_{i'} q_{i',S}}}{q_{i,S}}$ over all outcomes $S\subseteq M$, where the numerator is the likelihood given~$S$ of the weighted combination distribution $\sum_{i'} {\alpha_{i'} q_{i'}}$, and the denominator is the likelihood given~$S$ of distribution~$q_i$. Note that the weighted combination distribution is \emph{not} in general a product distribution itself.

Denote the optimal solution (i.e.~the minimum likelihood ratio) by~$\rho^*$. 
Solving the separation oracle problem is NP-hard (Proposition \ref{pro:separation-NP-hard}),%
\footnote{In fact the problem is strongly NP-hard; but because it involves products of the form $q_{i,S}=(\prod_{j\in S}q_{i,j})(\prod_{j\notin S}(1-q_{i,j}))$, the strong NP-hardness does not rule out an FPTAS \cite[Theorem 17.12]{PapadimitriouS82}.} 
but in Appendix~\ref{appx:separation-oracle} we show an FPTAS (Lemma~\ref{lem:FPTAS}). Lemma~\ref{lem:sep-FPTAS} gives the guarantee from applying this FPTAS as a separation oracle for dual LP~\eqref{LP:dual}. 

\begin{lemma}[FPTAS] 
	\label{lem:FPTAS}
	There is an algorithm for the separation oracle problem that returns an outcome $S$ with likelihood ratio $\le (1+\delta)\rho^*$ in time polynomial in $m,\frac{1}{\delta}$.
\end{lemma}

\begin{lemma} 
	\label{lem:sep-FPTAS}
	If the separation oracle FPTAS with parameter $\delta$ does not find a violated constraint of dual LP~\eqref{LP:dual}, then for every $S$ the inequality in \eqref{eq:dual-constraint} holds approximately up to $(1+\delta)$: 
	$$
	(\sum_{i'\ne i} {\lambda_{i'}}) - 1 \le (1+\delta)\sum_{i'\ne i} {\lambda_{i'}\frac{q_{i',S}}{q_{i,S}}}.
	$$
\end{lemma}

\begin{proof}
	Assume there exists $S$ such that $(\sum_{i'\ne i} {\lambda_{i'}}) - 1 > (1+\delta)\sum_{i'\ne i} {\lambda_{i'}\frac{q_{i',S}}{q_{i,S}}}$. Then dividing by $(\sum_{i'} {\lambda_{i'}})$ and using the definition of $\rho^*$ as the minimum likelihood ratio we get $1 - \frac{1}{\sum_{i'} {\lambda_{i'}}} > (1+\delta)\rho^*$. Combining this with the guarantee of Lemma \ref{lem:FPTAS}, the FPTAS returns $S'$ with likelihood ratio $<1 - \frac{1}{\sum_{i'} {\lambda_{i'}}}$, thus identifying a violated constraint. This completes the proof.
\end{proof}


\paragraph{Applying the separation oracle FPTAS: The standard method}
Given an FPTAS with parameter $\delta$ for the  separation oracle of a dual LP, for many problems it is possible to find in polynomial time an approximately-optimal, feasible solution to the primal---see, e.g., \cite{KarmarkarK82,CarrV02,JainMS03,NutovBY06,FleischerGMS11,FeldmanKN12}.
We first describe a fairly standard approach in the literature to utilizing a separation oracle FPTAS, which we refer to as the \emph{standard method}, and explain where we must deviate from this approach. The proof of Theorem~\ref{thm:FPTAS} then applies an appropriately modified approach.

The standard method works as follows: Let $OPT_i$ be the optimal value of the primal (minimization) LP. 
For a benchmark value $\Gamma$,
add to the (maximization) dual LP a constraint that requires its objective to be at least $\Gamma$, and attempt to solve the dual by running the ellipsoid algorithm with the separation oracle FPTAS. 

Assume first that the ellipsoid algorithm returns a solution with value $\Gamma$. Since the separation oracle applies the FPTAS, it may wrongly conclude that some solution is feasible despite a slight violation of one or more of the constraints. For example, if we were to apply the FPTAS separation oracle from Lemma \ref{lem:FPTAS} to solve dual LP~\eqref{LP:dual}, we could possibly get a solution for which there exists $S$ such that: 
$$
\sum_{i'\ne i} {\lambda_{i'}\frac{q_{i',S}}{q_{i,S}}} < (\sum_{i'\ne i} {\lambda_{i'}}) - 1 \le (1+\delta)\sum_{i'\ne i} {\lambda_{i'}\frac{q_{i',S}}{q_{i,S}}} 
$$
where the second inequality is by Lemma~\ref{lem:sep-FPTAS}.
Clearly, the value $\Gamma$ of an approximately-feasible solution may be higher than $OPT_i$. In the standard method, the approx-imately-feasible solution can be \emph{scaled} by $\frac{1}{1+\delta}$ to regain feasibility while maintaining value of $\frac{\Gamma}{1+\delta}$. Scaling thus establishes that $\frac{\Gamma}{1+\delta}\le OPT_i$. 
Now assume that for some (larger) value of $\Gamma$, the ellipsoid algorithm identifies that the dual LP is infeasible. In this case we can be certain that $OPT_i< \Gamma$, and we can also find in polynomial time a primal feasible solution with value $<\Gamma$ (more details in the proof of Theorem~\ref{thm:FPTAS} below).

Using binary search (in our case over the range $[c_i,R_i]\subseteq [0,1]$ since $R_i$ is the maximum the principal can pay without losing money), the standard method finds the smallest $\Gamma^*$ for which the dual is identified to be infeasible, up to a negligible binary search error $\epsilon$. 
This gives a primal feasible solution that achieves value $\Gamma^*+\epsilon$, and at the same time establishes that $\frac{(\Gamma^*)^-}{1+\delta}\le OPT_i$ by the scaling argument.%
\footnote{The notation $(\Gamma^*)^-$ means any number smaller than $\Gamma^*$.}
So the standard method has found an approximately-optimal, feasible solution to the primal.

\paragraph{Applying the separation oracle FPTAS: Our method}
The issue with applying the standard method to solve MIN-PAYMENT is that the scaling argument does not hold. To see this, consider an approximately-feasible dual solution for which 
$(\sum_{i'\ne i} {\lambda_{i'}}) - 1 \le (1+\delta)\sum_{i'\ne i} {\lambda_{i'}\frac{q_{i',S}}{q_{i,S}}}$ for every~$S$,
and notice that scaling the values $\{\lambda_{i'}\}$ 
does not achieve feasibility. 
We therefore turn to an alternative method to prove Theorem \ref{thm:FPTAS}. 

\begin{proof}[Proof of Theorem \ref{thm:FPTAS}]
We apply the standard method using the FPTAS with parameter $\delta$ (see Lemma~\ref{lem:FPTAS}) as separation oracle to the following \emph{strengthened} version of dual LP~\eqref{LP:dual},%
\footnote{Strengthened duals appear\final{, e.g.,}~in \cite{NutovBY06,FeldmanKN12}.} 
where the extra $(1+\delta)$ multiplicative factor in the constraints makes them harder to satisfy:
%
\begin{align}
\max~& \sum_{i'\ne i} {\lambda_{i'}(c_i-c_{i'})} & \label{LP:dual-tight}\\
\text{s.t.}~& (1+\delta)\big((\sum_{i'\ne i} {\lambda_{i'}}) - 1\big) \le \sum_{i'\ne i} {\lambda_{i'}\frac{q_{i',S}}{q_{i,S}}} && \forall S\subseteq E,q_{i,S}>0 \notag\\
& \lambda_{i'} \ge 0 &&\forall i'\ne i,i'\in [n].\notag
\end{align}

Let $\Gamma^*$ be the \final{infimum} 
value for which dual LP~\eqref{LP:dual-tight} would be identified as infeasible. The ellipsoid algorithm is thus able to find an approximately-feasible solution to dual LP~\eqref{LP:dual-tight} with objective $(\Gamma^*)^-$.
The key observation is that this solution is \emph{fully} feasible with respect to the original dual LP~\eqref{LP:dual}.
This is because if the separation oracle FPTAS does not find a violated constraint of dual LP~\eqref{LP:dual-tight}, then for every $S$ it holds that $(\sum_{i'\ne i} {\lambda_{i'}}) - 1 \le \sum_{i'\ne i} {\lambda_{i'}\frac{q_{i',S}}{q_{i,S}}}$ (by the same argument as in the proof of Lemma~\ref{lem:sep-FPTAS}). 
From the key observation it follows that 
\begin{equation}
(\Gamma^*)^-\le OPT_i\label{eq:key-obs}
\end{equation}
(despite the fact that the scaling argument does not hold). 

Now let $\Gamma^*+\epsilon$ be the smallest value for which the binary search runs the ellipsoid algorithm for dual LP~\eqref{LP:dual-tight} and identifies its infeasibility. During its run for $\Gamma^*+\epsilon$, the ellipsoid algorithm identifies polynomially-many separating hyperplanes that constrain the objective to $<\Gamma^*+\epsilon$. Formulate a ``small'' primal LP with variables corresponding exactly to these hyperplanes. By duality, the small primal LP has a solution with objective $<\Gamma^*+\epsilon$, and moreover since the number of variables and constraints is polynomial we can find such a solution $p^*$ in polynomial time. Observe that $p^*$ is also a feasible solution to the primal LP corresponding to dual~\eqref{LP:dual-tight} (the only difference from the small LP is more variables): 
%
\begin{align}
\min~& (1+\delta)\sum_{S\subseteq E} {q_{i,S} p_S} & \label{LP:min-pay-relax}\\
\text{s.t.}~& (1+\delta)\big(\sum_{S\subseteq E} {q_{i,S} p_S}\big) - c_i \ge \sum_{S\subseteq E} {q_{i',S} p_S} - c_{i'} &&\hspace*{-12pt}\forall i'\ne i,i'\in[n] \notag\\
& p_S \ge 0 &&\hspace*{-12pt}\forall S\subseteq E.\notag
\end{align}

We have thus obtained a contract $p^*$ that is a  feasible solution to LP~\eqref{LP:min-pay-relax} with objective $(1+\delta)\sum_{S\subseteq E} {q_{i,S} p_S} < \Gamma^*+\epsilon$. 
For action $a_i$, this contract pays the agent an expected transfer of $\sum_{S\subseteq E} {q_{i,S} p_S}< \frac{\Gamma^*+\epsilon}{1+\delta}$.
We have the following chain of inequalities:
$\sum_{S\subseteq E} {q_{i,S} p_S}\le \frac{(\Gamma^*)^{-}+\epsilon}{1+\delta} \le \frac{OPT_i+\epsilon}{1+\delta}\le OPT_i$, where the second inequality is by~\eqref{eq:key-obs}, and the last inequality is by taking the binary search error to be sufficiently small.%
\footnote{We use here that $OPT_i\ge c_i$ and that the number of bits of precision is polynomial.} 
To complete the proof we must show that $p^*$ is $\delta$-IC. This holds since the constraints of LP~\eqref{LP:min-pay-relax} ensure that for every action $a_{i'}\ne a_i$, using the notation $p_i=\sum_{S\subseteq E} {q_{i,S} p_S}$, we have
$p_{i'} - c_{i'} \le (1+\delta)p_i - c_i \le p_i - c_i +\delta p_i \le p_i - c_i +\delta$ (the last inequality uses that $p_i\le R_i \le 1$ by normalization).
\end{proof} 

\section{Hardness of approximation}
\label{sec:hardness}

\final{In this section unlike the previous one, the number of actions is no longer assumed to be constant.}
We show a hardness of approximation result for optimal contracts, based on the known hardness of approximation for MAX-3SAT. In his landmark paper, 
\cite{Hastad01} shows that it is NP-hard to distinguish between a satisfiable MAX-3SAT instance, and one in which there is no assignment satisfying more than $7/8+\alpha$ of the clauses, where $\alpha$ is an arbitrarily-small constant (Theorems 5.6 and 8.3 in \cite{Hastad01}).
We build upon this to prove our main technical contribution stated in Theorem \ref{thm:hardness-of-approx}, which immediately leads to our main results for this section in Corollaries~\ref{cor:hardness-of-approx}-\ref{cor:hardness-of-approx-delta}. 

\begin{theorem}
	\label{thm:hardness-of-approx}
	Let $c\in \mathbb{Z},c\ge 3$ be an (arbitrarily large) constant integer. 
	Let $\epsilon,\Delta\in\mathbb{R},\epsilon>0,\Delta\in[0,\frac{1}{20^c}]$ be such that $\frac{\epsilon-2\Delta^{1/c}}{3}\in(0,\frac{1}{20}]$ and $(\frac{\epsilon-2\Delta^{1/c}}{3})^c$ is an (arbitrarily small) constant.
	Then it is NP-hard to determine whether a principal-agent setting has an IC contract extracting full expected welfare, or whether there is no $\Delta$-IC contract extracting $>\frac{1}{c}+\epsilon$ of the expected welfare. 
\end{theorem}


We present two direct implications of Theorem \ref{thm:hardness-of-approx}. First, Corollary \ref{cor:hardness-of-approx} applies to the OPT-CONTRACT problem, and states hardness of approximation within any constant of the optimal expected payoff by an IC contract. (A similar result can be shown for MIN-PAYMENT; see Appendix~\ref{appx:hardness-min-payment}.)

\begin{corollary}
	\label{cor:hardness-of-approx}
	For any constant $c\in \mathbb{R},c\ge 1$, it is NP-hard to approximate the optimal expected payoff achievable by an IC contract to within a multiplicative factor $c$. 
\end{corollary}

Corollary \ref{cor:hardness-of-approx} suggests that in order to achieve positive results, we may want to follow the approach of the CDW framework and relax IC to $\Delta$-IC. That is, instead of trying to compute in polynomial time an approximately-optimal IC contract, we should try to compute in polynomial time a $\Delta$-IC contract with expected payoff that is guaranteed to approximately exceed that of the optimal IC contract. 
The next corollary establishes a computational limitation on this approach: 
Corollary~\ref{cor:hardness-of-approx-delta} fixes a constant approximation factor $c$, and derives $\Delta$ for which a $c$-approximation by a $\Delta$-IC contract is NP-hard to find. 
(It is also possible to reverse the roles---fix $\Delta$ and derive a constant approximation factor for which NP-hardness holds.)
We shall complement this limitation with a positive result in Section \ref{sec:approx}. 


\begin{corollary}
	\label{cor:hardness-of-approx-delta}	
	For any constant $c\in\mathbb{R},c\ge 5$ and $\Delta \le(\frac{1}{4c})^c$, it is NP-hard to find a $\Delta$-IC contract that guarantees $>\frac{2}{c}OPT$, where $OPT$ is the optimal expected payoff achievable by an IC contract.%
	\footnote{The relevant hardness notion is more accurately FNP-hardness.}
\end{corollary}

\begin{proof}
The corollary follows from Theorem \ref{thm:hardness-of-approx} by setting $\epsilon=\frac{1}{c}$. 
\end{proof}

It also follows from Theorem \ref{thm:hardness-of-approx} and Corollary \ref{cor:hardness-of-approx-delta} that for every $c,\Delta$ as specified, it is NP-hard to approximate the optimal expected payoff achievable by a $\Delta$-IC contract to within a multiplicative factor $c/2$. That is, hardness of approximation also holds for $\delta$-OPT-CONTRACT.







In the remainder of the section we prove Theorem~\ref{thm:hardness-of-approx}. After a brief overview, Section~\ref{sub:gap-avg-sat} sets up some tools for the proof, in Section \ref{sub:two-actions} we focus on the special case of $c=2$, and in Section~\ref{sub:c-actions} we prove the more general statement for any constant $c$. 

\subsection{Proof overview}

It will be instructive to consider first a version of Theorem \ref{thm:hardness-of-approx} for the case of $c = 2$:

\begin{theorem}
	\label{thm:hardness-of-approx-2}
	Let $\epsilon,\Delta\in\mathbb{R},\epsilon>0,\Delta\in[0,\frac{1}{20^2}]$ be such that $\frac{\epsilon-2\Delta^{1/2}}{3}\in(0,\frac{1}{20}]$ and 
	$(\frac{\epsilon-2\Delta^{1/2}}{3})^2$ is an (arbitrarily small) constant.
	Then it is NP-hard to determine whether a principal-agent setting has an IC contract extracting full expected welfare, or whether there is no $\Delta$-IC contract extracting $>\frac{1}{2}+\epsilon$ of the expected welfare. 
\end{theorem}

This theorem is already interesting as it shows that even relaxing IC to $\Delta$-IC where $\Delta\gg 0$, approximating the optimal expected payoff within $65\%$ is computationally hard:

\begin{corollary}
	For any $\Delta\le \frac{1}{20^2}$, 
	it is NP-hard to find a $\Delta$-IC contract that guarantees $>0.65 \cdot OPT$, 
	where $OPT$ is the optimal expected payoff achievable by an IC contract. 
\end{corollary}

\begin{proof}
	The corollary follows from Theorem \ref{thm:hardness-of-approx-2} by setting $\epsilon=\frac{3}{20}$. 
\end{proof}

To establish Theorem \ref{thm:hardness-of-approx-2} we present a gap-preserving reduction from any MAX-3SAT instance~$\varphi$ to a principal-agent setting that we call the ``product setting'' (the reduction appears in Algorithm~\ref{alg:sat-to-prod-2} and is analyzed in Proposition~\ref{pro:hardness-reduction-2}). 
The product setting encompasses a 2-action principal-agent ``gap setting'', in which any $\delta$-IC contract for sufficiently small $\delta$ cannot extract much more than $\frac{1}{2}$ of the expected welfare (Proposition \ref{pro:gap-setting-two-actions}). 

The special case of $c= 2$ captures most ideas behind the proof of the more general Theorem~\ref{thm:hardness-of-approx}, but the analysis is simplified by the fact that to extract more than roughly~$\frac{1}{2}$ of the expected welfare in the 2-action gap setting, there is a single action that the contract could potentially incentivize. The more general case involves gap settings with more actions (the reduction appears in Algorithm~\ref{alg:sat-to-prod} and is analyzed in Proposition~\ref{pro:hardness-reduction}). To extract more than $\approx\frac{1}{c}$ of the expected welfare, the contract could potentially incentivize almost any one of these actions (Proposition \ref{pro:gap-setting}).

\paragraph{Barrier to going beyond constant $\mathbf{c}$} 
Our techniques for establishing Theorem \ref{thm:hardness-of-approx} do not generalize beyond constant values of $c$ (the approximation factor). The reason for this is that we do not know of $(c,\epsilon,f)$-gap settings (Definition~\ref{def:unstruct-gap-setting}) where $f(c,\epsilon)=o(\epsilon^c)$. 
As long as $f(c,\epsilon)$ is of order $\epsilon^c$, the gap in the MAX-3SAT instance we reduce from must be between $7/8 + \epsilon^c$ and 1, and this gap problem is known to be NP-hard only for constant $c$. As \cite{Hastad01} notes, significantly stronger complexity assumptions may lead to hardness for slightly (but not significantly) larger values of $c$.

\subsection{Key ingredients}
\label{sub:gap-avg-sat}

In this section we formalize the notions of ``gap'' and ``SAT'' principal-agent settings as well as the notion of an ``average action'', which will be useful in proving Theorems \ref{thm:hardness-of-approx} and \ref{thm:hardness-of-approx-2}. 
The term ``gap setting'' reflects the gap between the first-best solution (i.e., the expected welfare), and the second-best solution (i.e., the expected payoff to the principal from the optimal contract). 
It will be convenient \emph{not} to normalize gap settings (and thus also the product settings encompassing them). This makes our negative results only stronger, as we show next.

\paragraph{Unnormalized settings and a stronger $\delta$-IC notion}
Before proceeding we must define what we mean by a $\delta$-IC contract in an unnormalized setting. Moreover we show that if Theorems~\ref{thm:hardness-of-approx} or \ref{thm:hardness-of-approx-2} hold for unnormalized settings with the new $\delta$-IC notion, then they also hold for normalized settings with the standard $\delta$-IC notion.

Recall that in a \emph{normalized} setting, action $a_i$ that is $\delta$-incentivized by the contract must satisfy $\delta$-IC constraints of the form $p_i-c_i+\delta \ge p_{i'}-c_{i'}$ for every $i'\ne i$. In an \emph{unnormalized} setting, an additive $\delta$-deviation from optimality is too weak of a requirement; we require instead that $a_i$ satisfy $\delta$-IC constraints of the form 
\begin{equation}
(1+\delta)p_i-c_i \ge p_{i'}-c_{i'}~~~\forall i'\ne i.\label{eq:delta-IC-orig}
\end{equation}
Two key observations are: (i) The constraints in \eqref{eq:delta-IC-orig} imply the standard $\delta$-IC constraints if $p_i\le 1$, as is the case if the setting is normalized; (ii) The constraints in \eqref{eq:delta-IC-orig} are invariant to scaling of the setting and contract (i.e., to a change of currency of the rewards, costs and payments). By these observations, a $\delta$-IC contract according to the new notion in an unnormalized setting becomes a standard $\delta$-IC contract after normalization of the setting and payments, with the same fraction of optimal expected welfare extracted as payoff to the principal.

Assume a negative result holds for unnormalized settings, i.e., it is NP-hard to determine between the two cases stated in Theorem~\ref{thm:hardness-of-approx} (or Theorem~\ref{thm:hardness-of-approx-2}). 
Assume for contradiction this does not hold for normalized settings. Then given an unnormalized setting, we can simply scale the expected rewards and costs to normalize it, and then determine whether or not there is an IC contract extracting full expected welfare. If such a contract exists, it is also IC and full-welfare-extracting in the unnormalized setting after scaling back the payments. 
On the other hand, by the discussion above, if there is no standard-notion $\Delta$-IC contract extracting a given fraction of the expected welfare in the normalized setting, there can also be no such contract with the new $\Delta$-IC notion in any scaling of the setting. We have this reached a contradiction to NP-hardness. We conclude that proving our negative results for unnormalized settings only strengthens these results.

\paragraph{Gap settings and their construction} 
We now turn to the definition of gap settings.

\begin{definition}[Unstructured gap setting]
	\label{def:unstruct-gap-setting}
	Let $f(c,\epsilon)\in \mathbb{R}_{\ge 0}$ be an increasing function where $c\in \mathbb{Z}_{> 0}$ and $\epsilon\in \mathbb{R}_{> 0}$.
	An \emph{unstructured $(c,\epsilon,f)$-gap setting} is a principal-agent setting such that for every $0\le \delta\le f(c,\epsilon)$, the optimal $\delta$-IC contract can extract no more than $\frac{1}{c}+\epsilon$ of the expected welfare as the principal's expected payoff.
\end{definition}

For convenience we focus on (structured) gap settings as follows.

\begin{definition}[Gap setting]
	\label{def:gap-setting}
	 A \emph{$(c,\epsilon,f)$-gap setting} is a setting as in Definition \ref{def:unstruct-gap-setting} with the following structure: there is a~single item and $c$ actions; the first action has zero cost; the last action has probability 1 for the item and maximum expected welfare among all actions. 
\end{definition}

To construct a gap setting, we construct a principal-agent setting with a single item, $c$~actions and parameter \final{$\gamma\in \mathbb{R}_{>0}, \gamma<1$}.
The construction is similar to \cite{DRT18}, but requires a different analysis. 
For every $i\in[c]$, set the probability of action $a_i$ for the item to $\gamma^{c-i}$, and set $a_i$'s cost to 
$c_i=(1/\gamma^{i-1}) - i +(i-1)\gamma.$ 
Set the reward for the item to be $1/\gamma^{c-1}$.
Observe that the expected welfare of action $a_i$ is
$i -(i-1)\gamma$, so the last action has the maximum expected welfare $c-(c-1)\gamma$. 
This establishes the structural requirements from a gap setting (Definition~\ref{def:gap-setting}).
Propositions \ref{pro:gap-setting-two-actions} and \ref{pro:gap-setting} establish the gap requirements from a gap setting (Definition \ref{def:unstruct-gap-setting}) for $c=2$ and $c\ge 3$, respectively---the separation between these cases is for clarity of presentation. 
We use the former in Section~\ref{sub:two-actions}, in which we show hardness for the $c=2$ case; the latter is a generalization to arbitrary-large constant $c$. 
See Appendix~\ref{appx:gap-settings} for proofs. 

\begin{proposition}[$2$-action gap settings]
	\label{pro:gap-setting-two-actions}
	For every $\epsilon\in (0,\frac{1}{4}]$, there exists a $(2,\epsilon,\epsilon^2)$-gap setting. 
\end{proposition}

\begin{proposition}[$c$-action gap settings]
	\label{pro:gap-setting}
	For every $c\ge 3$ and $\epsilon\in (0,\frac{1}{4}]$, there exists a $(c,\epsilon,\epsilon^c)$-gap setting.  
\end{proposition}


For concreteness we describe the 2-action gap setting: The agent has $c=2$ actions, which can be thought of as ``effort'' and ``no effort''. Effort has cost $\frac{1}{\epsilon}-2+\epsilon$, and no effort has cost $0$. Without effort the item has probability $\epsilon$, and with effort the probability is $1$. The reward associated with the item is $\frac{1}{\epsilon}$. It is immediate to see that the maximum expected welfare (first-best) is $2-\epsilon$. In the proof of Proposition~\ref{pro:gap-setting-two-actions} we show that the best an $\epsilon^2$-IC contract can extract is $\approx 1$.  

\paragraph{Average actions and SAT settings}
The motivation for the next definition is that given a contract, for an action to be IC or $\delta$-IC it must yield higher expected utility for the agent in comparison to the ``average action''. Average actions are thus a useful tool for analyzing contracts. 

\begin{definition}[Average action]
	Given a principal-agent setting and a subset of actions, by the \emph{average action} we refer to a hypothetical action with the average of the subset's distributions, and average cost. (If a particular subset is not specified, the average is taken over all actions in the setting.)
\end{definition}

Another useful ingredient will be SAT settings defined as follows. 

\begin{definition}[SAT setting]
	\label{def:sat-setting}
	A \emph{SAT} principal-agent setting corresponds to a MAX-3SAT instance $\varphi$. If $\varphi$ has $n$ clauses and $m$ variables then the SAT setting has $n$ actions and $m$ items. Two conditions hold:
	(1) $\varphi$ is satisfiable if and only if there is an item set in the SAT setting that the average action leads to with zero probability;
	(2) If every assignment to $\varphi$ satisfies at most $7/8+\alpha$ of the clauses, then for every item set $S$ the average action leads to $S$ with probability at least $\frac{1-8\alpha}{2^m}$. 
\end{definition}

The following proposition (whose proof appears in Appendix~\ref{appx:sat-setting}) provides a reduction from MAX-3SAT instances to SAT settings.

\begin{proposition} 
	\label{pro:sat-setting}
	For every $\varphi$ the reduction in Algorithm \ref{alg:sat-setting-from-phi} runs in polynomial time on input $\varphi$ and returns a SAT setting corresponding to $\varphi$.
\end{proposition}

\begin{algorithm}[ht]
	\Input{A MAX-3SAT instance $\varphi$ with $n$ clauses and $m$ variables.}
	\Output{A principal-agent SAT setting (Definition \ref{def:sat-setting}) corresponding to $\varphi$.}
	\Begin{
		Given $\varphi$, construct a principal-agent setting in which every clause corresponds to an action with a product distribution, and for every variable there is a corresponding item. If variable $j$ appears in clause $i$ of $\varphi$ as a positive literal, then let item $j$'s probability in the $i$th product distribution be 0, and if it appears as a negative literal then let item $j$'s probability be 1. Set all other probabilities to be $\frac{1}{2}$. We set the costs of all actions and the rewards for all items to be $0$. 
	}
	\caption{SAT setting construction in polytime}
	\label{alg:sat-setting-from-phi}
\end{algorithm}

\subsection{The $\mathbf{c=2}$ case: Proof of Theorem \ref{thm:hardness-of-approx-2}}
\label{sub:two-actions}

In this section we present a polynomial-time reduction from MAX-3SAT to a product setting, which combines gap and SAT settings. The reduction appears in Algorithm~\ref{alg:sat-to-prod-2}.
We then analyze the guarantees of the reduction and use them to prove Theorem \ref{thm:hardness-of-approx-2}. Most of the analysis appears in Proposition~\ref{pro:hardness-reduction-2}, which shows that the reduction in Algorithm \ref{alg:sat-to-prod-2} is gap-preserving.
Some of the results are formulated in general terms so they can be reused in the next section (Section \ref{sub:c-actions}). 

Before turning to Proposition \ref{pro:hardness-reduction-2}, we begin with two simple observations about the product setting resulting from the reduction. 

\begin{algorithm}[ht]
\Input{A MAX-3SAT instance $\varphi$ with $n$ clauses and $m$ variables; a parameter $\epsilon\in\mathbb{R}_{\ge0}$.}
\Output{A principal-agent \emph{product setting} combining a \emph{SAT setting} and a \emph{gap setting}.}
\Begin{
Combine the SAT setting corresponding to $\varphi$ (attainable in polytime by Proposition \ref{pro:sat-setting}) with a poly-sized $(2,\epsilon,\epsilon^2)$-gap setting (exists by Proposition \ref{pro:gap-setting-two-actions}) to get the product setting, as follows:
\begin{minipage}{0.946\linewidth}%
\begin{itemize}
	\item The product setting has $n+1$ actions and $m+1$ items: $m$ ``SAT items'' correspond to the SAT setting items, and the last ``gap item'' corresponds to  the gap setting item. 
	
	\item The upper-left block of the product setting's $(n+1) \times(m+1)$ matrix of probabilities is the SAT setting's $n\times m$ matrix of probabilities. The entire lower-left $1\times m$ block is set to $\frac{1}{2}$. The entire upper-right $n \times 1$ block is set to the probability that action $a_1$ in the gap setting results in the item. The remaining lower-right $1 \times 1$ block is set to the probability that the last action in the gap setting results in the item.
	
	\item In the product setting, the rewards for the $m$ SAT items are set to 0, and the reward for the gap item is set as in the gap setting. 
	
	\item The costs of the first $n$ actions in the product setting are the cost of action $a_1$ in the gap setting; the cost of the last action in the product setting is the  cost of the last action in the gap setting.  
\end{itemize}
\end{minipage}
}
\caption{Polytime reduction from MAX-3SAT to principal-agent}
\label{alg:sat-to-prod-2}
\end{algorithm}

\begin{observation}
\label{obs:product-rewards}
Partition all actions of the product setting but the last one into blocks of $n$ actions each.%
\footnote{If the number of actions in the gap setting is $2$, there is a single such block.} 
Every action in the $i$th block has the same expected reward for the principal as action $a_i$ in the gap setting, and the last action in the product setting has the same expected reward as the last action in the gap setting. 
\end{observation}

\begin{corollary}
\label{cor:product-welfares}
The optimal expected welfares of the product and gap settings are the same, and are determined by their respective last actions.
\end{corollary}

\begin{proposition}[Gap preservation by Algorithm~\ref{alg:sat-to-prod-2}]
	\label{pro:hardness-reduction-2} 
	Let $\varphi$ be a MAX-3SAT instance for which either there is a satisfying assignment, or every assignment satisfies at most $7/8+\alpha$ of the clauses for $\alpha\le (0.05)^2$. 
	Let $\Delta\le (0.05)^2$.
	Consider the product setting resulting from the reduction in Algorithm~\ref{alg:sat-to-prod-2} run on input $\varphi,\epsilon=3\alpha^{1/2}+2\Delta^{1/2}\le \frac{1}{4}$. Then:
	\begin{enumerate}
		\item If $\varphi$ has a satisfying assignment, the product setting has an IC contract that extracts full expected welfare;
		\item If every assignment to $\varphi$ satisfies at most $7/8+\alpha$ of the clauses, the optimal $\Delta$-IC contract can extract no more than $\frac{1}{2}+\epsilon$ of the expected welfare.
	\end{enumerate}
\end{proposition}

\begin{proof}
First, if $\varphi$ has a satisfying assignment, then there is a subset of SAT items that has zero probability according to every one of the first $n$ actions. Consider the outcome $S^*$ combining this subset together with the gap item. We construct a full-welfare extracting contract: the contract's payment for $S^*$ is the cost of the last action in the product setting multiplied by $2^m$ (since the probability of $S^*$ according to the last action is $1/2^m$), and all other payments are set to zero. It is not hard to see that the resulting contract makes the agent indifferent among all actions, so by tie-breaking in favor of the principal, the principal receives the full expected welfare as her payoff. 

Now consider the case that every assignment to $\varphi$ satisfies at most $7/8+\alpha$ of the clauses, and assume for contradiction that there is a $\Delta$-IC contract $p$ for the product setting that extracts more than $\frac{1}{2}+\epsilon$ of the expected welfare. 
We derive from $p$ a $\delta$-IC contract $p'$ for the $(2,\epsilon,\epsilon^2)$-gap setting where $\delta\le \epsilon^2$, which extracts more than $\frac{1}{2} + \epsilon$ of the expected welfare. This is a contradiction to the properties of the gap setting (Definition \ref{def:unstruct-gap-setting}). 

It remains to specify and analyze contract $p'$ : 
For brevity we denote the singleton containing the gap item by $M'$, and define 
\begin{eqnarray}
&p'(S') = \frac{1-8\alpha}{2^m} \sum_{S\subseteq [m]} p(S\cup S') & \forall S'\subseteq M',\label{eq:p2p}
\end{eqnarray}
where $S'$ is either the singleton containing the gap item or the empty set.
The starting point of the analysis is the observation that to extract $>\frac{1}{2}+\epsilon$ of the expected welfare in the product setting, contract $p$ must $\Delta$-incentivize the last action (this follows since the expected rewards and costs of the actions are as in the gap setting by Observation~\ref{obs:product-rewards}, and so the same argument as in the proof of Proposition \ref{pro:gap-setting-two-actions} holds). 

Claim~\ref{cla:last-action} below establishes that if contract $p$ $\Delta$-incentivizes the last action in the product setting, then contract $p'$ $\delta$-incentivizes the last action in the gap setting for $\delta=\frac{8\alpha+\Delta}{1-8\alpha}$. So indeed
\begin{align*}
\delta &=\frac{8\alpha}{1-8\alpha}+\frac{\Delta}{1-8\alpha}\\
&\le 9\alpha+4\Delta\\
&=(3\alpha^{1/2})^2+(2\Delta^{1/2})^2\\
&\le (3\alpha^{1/2} + 2\Delta^{1/2})^2 \quad = \epsilon^2,
\end{align*}
using that $\alpha,\Delta\le (0.05)^2$ for the first inequality.

Now observe that the expected payoff to the principal from contract $p'$ $\delta$-incen-tivizing the last gap setting action is at least that of contract $p$ $\Delta$-incentivizing the last product setting action: the payments of $p'$ as defined in \eqref{eq:p2p} are the average payments of $p$ lowered by a factor of $(1-8\epsilon)$, and the expected rewards in the two settings are the same (Observation~\ref{obs:product-rewards}). 
The expected welfares in the two settings are also equal (Corollary \ref{cor:product-welfares}). We conclude that like contract $p$ in the product setting, contract $p'$ guarantees extraction of $>\frac{1}{2} + \epsilon$ of the expected welfare in the gap setting.
This leads to a contradiction and completes the proof of Proposition \ref{pro:hardness-reduction-2} (up to Claim \ref{cla:last-action} proved below).
\end{proof}

The next claim is formulated in general terms so that it can also be used in Section \ref{sub:c-actions}. It references the contract $p'$ defined in \eqref{eq:p2p}.

\begin{claim}
	\label{cla:last-action}
	Assume every assignment to the MAX-3SAT instance $\varphi$ satisfies at most $7/8+\alpha$ of its clauses where $\alpha<\frac{1}{8}$, and consider the product and gap settings returned by the reduction in Algorithm \ref{alg:sat-to-prod-2} (resp., Algorithm \ref{alg:sat-to-prod}). If in the product setting the last action is $\Delta$-incentivized by contract $p$, then in the gap setting the last action is $\delta$-incentivized by contract $p'$ for $\delta=\frac{8\alpha+\Delta}{1-8\alpha}$. 
\end{claim}

\begin{proof}
Let $g_i$ denote the distribution of action $a_i$ in the gap setting and let $c$ be the number of actions in this setting. 
In the product setting, by construction its last action assigns probability $\frac{g_c(S')}{2^m}$ to every set $S\cup S'$ such that $S$ contains SAT items and $S'\subseteq M'$.
Thus the expected payment for the last action given contract~$p$ is
\begin{align}
&\sum_{S\subseteq [m]} \sum_{S'\subseteq M'} \frac{g_c(S')}{2^m} p(S\cup S') = \frac{1}{1-8\alpha}\sum_{S'\subseteq M'} g_c(S')p'(S'), \label{eq:pay-last}
\end{align}
where the equality follows from the definition of $p'$ in \eqref{eq:p2p}. Note that the resulting expression in \eqref{eq:pay-last} is precisely the expected payment for the last action in the gap setting given contract $p'$, multiplied by factor $1/(1-8\alpha)$. 

Similarly, for every $i\in c$ consider the average action over the $i$th block of $n$ actions in the product setting.%
\footnote{If $c=2$ there is a single such block.} 
Again by construction, the probability this $i$th average action
assigns to $S\cup S'$ is $\ge \frac{g_i(S')(1-8\alpha)}{2^m}$, where we use that the average action of the SAT setting has probability $\ge \frac{1-8\alpha}{2^m}$ for $S$ (Definition \ref{def:sat-setting}). Thus the expected payment for the $i$th average action given contract $p$ is at least
\begin{align}
&\sum_{S\subseteq [m]}\sum_{S'\subseteq M'} \frac{g_i(S')(1-8\alpha)}{2^m} p(S\cup S') = \sum_{S'\subseteq M'} g_i(S')p'(S')& \forall i\in[c],\label{eq:pay-first}
\end{align}
where again the equality follows from \eqref{eq:p2p}. Note that the resulting expression in \eqref{eq:pay-first} is precisely the expected payment for action $a_i$ in the gap setting given contract $p'$.

We now use the assumption that in the product setting, contract $p$ $\Delta$-incentivizes the last action. This means the agent $\Delta$-prefers the last action to the $i$th average action, which has cost zero. Combining \eqref{eq:pay-last} and \eqref{eq:pay-first} we get
\begin{align}
&\frac{1+\Delta}{1-8\alpha}\sum_{S'\subseteq M'} g_c(S')p'(S') - \mathcal{C} \ge \sum_{S'\subseteq M'} g_i(S')p'(S')& \forall i\in[c], \label{eq:delta-IC}
\end{align}
where $\mathcal{C}$ denotes the cost of the last action in the product and gap settings. By definition of $\delta$-IC, Inequality \eqref{eq:delta-IC} immediately implies that in the gap setting, the last action is $\delta$-IC given contract $p'$ where $\delta= \frac{8\alpha+\Delta}{1-8\alpha}$, thus completing the proof of Claim \ref{cla:last-action}.
\end{proof}

We can now use Proposition \ref{pro:hardness-reduction-2} to prove Theorem~\ref{thm:hardness-of-approx-2}. 

\begin{proof}[Proof of Theorem~\ref{thm:hardness-of-approx-2}]
	Recall that $\frac{(\epsilon-2\Delta^{1/2})^2}{9}$ is a constant $\le (0.05)^2$.
	Assume a polynomial-time algorithm for determining whether a principal-agent setting has a (fully-IC) contract that extracts the full expected welfare, or whether no $\Delta$-IC contract can extract more than $\frac{1}{2}+\epsilon$. Then given a MAX-3SAT instance $\varphi$ for which either there is a satisfying assignment or every assignment satisfies at most $\frac{7}{8}+\frac{(\epsilon-2\Delta^{1/2})^2}{9}$ of the clauses, by Proposition~\ref{pro:hardness-reduction-2} the product setting (constructed in polynomial time) either has a full-welfare extracting contract or has no $\Delta$-IC contract that can extract more than $\frac{1}{2}+\epsilon$. Since the algorithm can determine among these two cases, it can solve the MAX-3SAT instance $\varphi$. But by \cite{Hastad01} and since $\frac{(\epsilon-2\Delta^{1/2})^2}{9}$ is a constant, we know that there is no polynomial-time algorithm for solving such MAX-3SAT instances unless $P=NP$. This completes the proof of Theorem \ref{thm:hardness-of-approx-2}.
\end{proof} 

\subsection{The general case: Proof of Theorem \ref{thm:hardness-of-approx}}
\label{sub:c-actions}

In this section we formulate and analyze the guarantees of the reduction in Algorithm \ref{alg:sat-to-prod}.

\begin{algorithm}[ht]
	\Input{A MAX-3SAT instance $\varphi$ with $n$ clauses and $m$ variables; parameters $\epsilon\in\mathbb{R}_{\ge 0}$ and $c\in\mathbb{Z}_{>0}$ where $c\ge 3$.}
	\Output{A principal-agent \emph{product setting} combining copies of a \emph{SAT setting} and a \emph{gap setting}.}
	\Begin{
		Combine multiple copies of the SAT setting corresponding to $\varphi$ (attainable in polytime by Proposition~\ref{pro:sat-setting}) with a poly-sized $(c,\epsilon,\epsilon^c)$-gap setting (exists by Proposition \ref{pro:gap-setting}) to get the product setting, as follows:
		\begin{minipage}{0.946\linewidth}%
		\begin{itemize}
			\item The product setting has $cn+1$ actions and $m+1$ items: $m$ ``SAT items'' correspond to the SAT setting items, and the last ``gap item'' corresponds to the gap setting item. 
			
			\item For every $i\in [c]$, consider the $i$th block of $n$ rows of the product setting's $(cn+1)\times (m+1)$ matrix of probabilities. The $i$th block consists of row $(i-1)\cdot n+1$ to row $i\cdot n$ and forms a submatrix of size $n\times (m+1)$. The first $m$ columns of the sub-matrix are set to a copy of the SAT setting's $n\times m$ matrix of probabilities, and the entire last column is set to the probability that action $a_i$ in the gap setting results in the item. 		
			Finally, the first $m$ entries of the last row of the product setting's matrix (i.e., row $cn+1$) are set to $\frac{1}{2}$, and the last entry (the lower-right corner of the matrix) is set to the probability that the last action in the gap setting results in the item. 
			
			\item In the product setting, the rewards for the $m$ SAT items are set to 0, and the reward for the gap item is set as in the gap setting. 
			
			\item For every $i\in [c]$, the costs of the $n$ actions in block $i$ are the cost of action $a_i$ in the gap setting; the cost of the last action in the product setting is the cost of the last action in the gap setting.  
		\end{itemize}
		\end{minipage}
	}
	\caption{Generalized polytime reduction from MAX-3SAT to principal-agent}
	\label{alg:sat-to-prod}
\end{algorithm}

\begin{proposition}[Gap preservation by Algorithm \ref{alg:sat-to-prod}]
	\label{pro:hardness-reduction}
	Let $c\in\mathbb{Z},c\ge 3$. 
	Let $\varphi$ be a MAX-3SAT instance for which either there is a satisfying assignment, or every assignment satisfies at most $7/8+\alpha$ of the clauses for $\alpha\le (0.05)^c$. 
	Let $\Delta\le (0.05)^c$.
	Consider the product setting resulting from the reduction in Algorithm~\ref{alg:sat-to-prod} run on input $\varphi,c,\epsilon=3\alpha^{1/c}+2\Delta^{1/c} \le \frac{1}{4}$. Then:
	\begin{enumerate}
		\item If $\varphi$ has a satisfying assignment, the product setting has an IC contract that extracts full expected welfare;
		\item If every assignment to $\varphi$ satisfies at most $7/8+\alpha$ of the clauses, the optimal $\Delta$-IC contract can extract no more than $\frac{1}{c}+\epsilon$ of the expected welfare.
	\end{enumerate}
\end{proposition}

\begin{proof}
	First, if $\varphi$ has a satisfying assignment, then there is a subset of SAT items that has zero probability according to every one of the actions in the product setting except for the last action, and so we can construct a full-welfare extracting contract as in the proof of Proposition \ref{pro:hardness-reduction-2}. From now on consider the case that every assignment to $\varphi$ satisfies at most $7/8+\alpha$ of the clauses, and assume for contradiction there is a $\Delta$-IC contract $p$ for the product setting that extracts more than $\frac{1}{c}+\epsilon$ of the expected welfare. 
	
	Consider the case that $p$ $\Delta$-incentivizes the last action in the product setting. Then we can derive from it a $\delta$-IC contract $p'$ for the $(c,\epsilon,\epsilon^c)$-gap setting where $\delta\le \epsilon^c$, which extracts more than $\frac{1}{c} + \epsilon$ of the expected welfare. 
	This is a contradiction to the properties of the gap setting (Definition \ref{def:unstruct-gap-setting}). The construction of $p'$ and its analysis are as in the proof of Proposition \ref{pro:hardness-reduction-2} (where Equation \eqref{eq:p2p} defines $p'$), and so are omitted here except for the following verification: we must verify that indeed $\delta\le \epsilon^c$. We know from Claim \ref{cla:last-action} that $\delta=\frac{8\alpha+\Delta}{1-8\alpha}$. As in the proof of Proposition \ref{pro:hardness-reduction-2} this is $\le 9\alpha+4\Delta$, and it is not hard to see that
	$$
	9\alpha+4\Delta \le 
	(3\alpha^{1/c})^c+(2\Delta^{1/c})^c \le
	(3\alpha^{1/c}+2\Delta^{1/c})^c =
	\epsilon^c,
	$$
	where the first inequality uses that $c\ge 3$. 

	In the remaining case, $p$ $\Delta$-incentivizes an action $a_{i^*k}$ in the product setting which is the $k$th action in block $i^*\in[c]$ (recall each block has $n$ actions). We derive from $p$ a contract $p'_k$ (depending on $k$) for the gap setting that $\Delta$-incentivizes $a_{i^*}$ at the same expected payment. 
	As in the proof of Proposition \ref{pro:hardness-reduction}, this means that $p'_k$ extracts $>\frac{1}{c}+\epsilon$ of the expected welfare in the gap setting.
	Since $\Delta\le \delta=\frac{8\alpha+\Delta}{1-8\alpha}$ it follows from the argument above that $\Delta\le \epsilon^c$, and so we have reached a contradiction to the properties of the gap setting (Definition \ref{def:unstruct-gap-setting}).
	
	We define $p'_k$ as follows: Let $s_k$ denote the distribution of action $a_k$ in the SAT setting. For every subset $S'\subseteq M'$ of gap items,
	\begin{align}
	&p'_k(S') = \sum_{S\subseteq [m]} p(S\cup S') s_k(S)&& \forall S'\subseteq M',\label{eq:p2p-more-actions}
	\end{align}
	where $S'$ is either the singleton containing the gap item or the empty set.
	
	For the analysis, let $g_i$ denote the distribution of action $a_i$ in the gap setting. 
	In the product setting, for every $i\in[c],k\le n$ the expected payment for action $a_{ik}$ by contract $p$ is 
	\begin{equation}
	\sum_{S\in [m]}	\sum_{S'\subseteq M'} s_k(S)g_i(S') p(S\cup S').\label{eq:p-more-actions-analysis}
	\end{equation}
	In the gap setting, the expected payment for $a_{i}$ by contract $p'_k$ is
	$\sum_{S'\subseteq M'}g_i(S')p'(S')$, and by definition of $p'_k$ in~\eqref{eq:p2p-more-actions} this coincides with the expected payment in~\eqref{eq:p-more-actions-analysis}. We know that contract $p$ $\Delta$-incentivizes $a_{i^*k}$ in the product setting, in particular against any action $a_{ik}$ where $i\in [c]\setminus\{i^*\}$ (i.e., against actions in the same position $k$ but in different blocks). This implies that contract $p'_k$ $\Delta$-incentivizes $a_{i^*}$ in the gap setting against any action $a_i$, completing the proof.
\end{proof}

We can now use Proposition \ref{pro:hardness-reduction} to prove Theorem \ref{thm:hardness-of-approx}. The proof is identical to that of Theorem~\ref{thm:hardness-of-approx-2} and so is omitted here.

\section{Approximation guarantees}
\label{sec:approx}

In this section we show that for any constant $\delta$ there is a simple, namely linear, $\delta$-IC contract that extracts as expected payoff for the principal a $c_\delta$-fraction of the optimal welfare, where $c_\delta$ is a constant that depends only on $\delta$. Recall that a linear contract is defined by a parameter $\alpha \in [0,1]$, and pays the agent $p_S = \alpha \sum_{j \in S} r_j$ for every outcome $S \subseteq M$.

\begin{theorem}\label{thm:delta-ic-approx}
Consider a principal-agent setting with $n$ actions. For every $\gamma \in (0,1)$ and every $\delta > 0$ there is a $\delta$-IC linear contract with expected payoff $ALG$ where
\[
\text{ALG} \geq \left((1-\gamma) \frac{1}{\lceil \log_{1+\delta}(\frac{1}{\gamma})\rceil + 1}\right) \max_{i \in [n]} \{R_i - c_i\}.
\]
\end{theorem}
An immediate corollary of Theorem~\ref{thm:delta-ic-approx} is 
that we can compute a $\delta$-IC linear contract that achieves the claimed \final{constant-factor} approximation in polynomial time.
By Corollary~\ref{cor:hardness-of-approx} we cannot achieve a similar result for IC (rather than $\delta$-IC) \final{contracts} unless $P=NP$. 
In fact, an even stronger lower bound holds for \final{the class of exactly IC linear} \textcolor{black}{(or, more generally, separable)} \final{contracts. These contracts cannot achieve an approximation ratio better than $n$ (see \cite{DRT18}} and Appendix~\ref{appx:separable} \final{for details).}

\paragraph{Geometric understanding of linear contracts}
To prove Theorem~\ref{thm:delta-ic-approx} we will rely on the following geometric understanding of linear contracts developed in \cite{DRT18}. Fix a principal-agent setting. For a linear contract with parameter $\alpha \in[0,1]$ and an action~$a_i$, the expected reward $R_i = \sum_S q_{i,S} r_S$ is split between the principal and the agent, leaving the principal with $(1-\alpha) R_i$ in expected utility and the agent with $\alpha R_i - c_i$ (the sum of the players' expected utilities is action $a_i$'s expected welfare). The agent's expected utility for choosing action $a_i$ \emph{as a function of $\alpha$} is thus a line from $-c_i$ (for $\alpha=0$) to $R_i - c_i$ (for $\alpha=1$). Drawing these lines for each of the $n$ actions, we trace 
the agent's utility for his best action as $\alpha$ goes from $0$ to $1$. This gives us the \emph{upper envelope} diagram for linear contracts in the given principal-agent setting. 


We now analyze some properties of the actions along the upper envelope diagram (i.e., as the linear contract gives an increasingly higher fraction of the rewards to the agent). 
Let $I_N$ be the subset of $N \leq n$ actions implementable by some linear contract. 
We subdivide the interval $[0,1]$ into $N \leq n$ intervals $T_1 = [\ell_1,r_1), T_2=[\ell_2,r_2), \dots, T_N = [\ell_N,r_N]$, with $\ell_1 = 0$, $\ell_i = r_{i-1}$ for $2 \leq i \leq N$, and $r_N = 1$. The subdivision is such that there is a bijection $\tau$ between the indices of actions $a_i \in I_N$ and between those of intervals $T_{\tau(i)}$, with the following properties:

\begin{enumerate}
\item For every $i\in[N]$ and $\alpha$ in the $i$th interval $T_i$, the linear contract with parameter $\alpha$ incentivizes action $a_{\tau^{-1}(i)}$. It follows that every action $a_i\in I_N$ appears on the upper envelope once, and the smallest $\alpha$ that incentivizes it is the left endpoint $\ell_{\tau(i)}$ of this action's interval. 
\item For every $i\in[N]$ it holds that $c_{\tau^{-1}(i)} \leq c_{\tau^{-1}(i+1)}$, $R_{\tau^{-1}(i)}  \leq R_{\tau^{-1}(i+1)}$, and $R_{\tau^{-1}(i)} - c_{\tau^{-1}(i)} \leq R_{\tau^{-1}(i+1)} - c_{\tau^{-1}(i+1)}$.   
\end{enumerate}

\paragraph{Notation}
\final{Our proof of Theorem~\ref{thm:delta-ic-approx} uses the following notation:}
 
\begin{itemize}
	\item We renumber the actions as they appear on the upper envelope from left to right. By the second property above, we get actions sorted by increasing cost, increasing expected reward and increasing expected welfare. In particular, the final action (action $a_N$ after renaming) must be the action in $A_n$ with the highest expected welfare.%
	\footnote{This is easy to see for $\alpha=1$, for which the full reward is transferred to the agent who also bears the cost, and so picks the welfare-maximizing action.}
  
	\item For every linearly-implementable action $a_i\in I_N$, we denote by $\alpha_i$ the smallest parameter $\alpha$ of a linear contract that incentivizes $a_i$ (i.e., the left endpoint of $a_i$'s corresponding interval).
\end{itemize}



\paragraph{Approximation guarantee proof}
\final{With these definitions at hand, we are now ready to prove the theorem.}


\begin{proof}[Proof of Theorem~\ref{thm:delta-ic-approx}] 
Given $\gamma \in (0,1)$ and $\delta > 0$, set $\kappa = \lceil \log_{1+\delta}(\frac{1}{\gamma}) \rceil$.
Subdivide the range $[0,1]$ of $\alpha$-parameters into $\kappa + 1$ intervals:
\begin{align*}
&[0,\gamma(1+\delta)^0), [\gamma(1+\delta)^0,\gamma(1+\delta)^1),\\&\hspace*{60pt}[\gamma(1+\delta)^1,\gamma(1+\delta)^2), \dots, 
[\gamma(1+\delta)^{\kappa-1},1].
\end{align*}

For each interval $k \in [\kappa+1]$, denote by 
$a_{h(k)}$ the action $a_i$ with the 
highest expected reward for which $\alpha_i$ falls into this interval
\textcolor{black}{(for simplicity of presentation we assume without loss of generality that such an action exists for each interval).}

Note that $h(k) < h(k+1)$ due to renaming and because actions appear on the upper envelope in non-decreasing order of expected reward.
We require the following definition: For $k \ge 2$, define 
$$
\alpha_{h(k-1),h(k)} = \frac{c_{h(k)}-c_{h(k-1)}}{R_{h(k)}-R_{h(k-1)}},
$$ 
i.e., $\alpha_{h(k-1),h(k)}$ is the $\alpha$ that makes the agent indifferent between actions $h(k-1)$ and $h(k)$.
For $k=1$ define $\alpha_{h(k-1),h(k)} = 0$. 

\final{The proof now proceeds by two claims.} 
\final{The first claim derives an upper bound on $\max_{i \in [n]}(R_i - c_i) = R_N -c_N$.}  

\begin{claim}
\label{cla:delta-ic-approx-1}
$
\max_{i \in [n]} (R_i - c_i) = R_N-c_N \leq \sum_{k = 1}^{\kappa+1} (1-\alpha_{{h(k-1)},{h(k)}}) R_{h(k)}.
$
\end{claim}

To prove Claim~\ref{cla:delta-ic-approx-1} we rely on the following observation from \cite{DRT18}. 

\begin{observation}\label{obs:observation-6}
Consider two actions $a_i, a_{i'}$ such that $a_i$ has higher expected reward and weakly higher welfare than $a_{i'}$, i.e., $R_{i} > R_{i'}$ and $R_i - c_i \geq R_{i'} - c_{i'}$, and let $\alpha_{i',i} = (c_i- c_{i'})/(R_{i} - R_{i'})$. 
Then
\[
	(R_i - c_i) - (R_{i'} - c_{i'}) \leq (1 - \alpha_{i',i})R_i.
\]
\end{observation}

\begin{myproof}[Proof of Claim~\ref{cla:delta-ic-approx-1}]
We argue by induction that for all $k \geq 1$, $R_{h(k)} - c_{h(k)} \leq \sum_{i = 1}^{k} (1-\alpha_{{h(i-1)},{h(i)}}) R_{h(i)}.$ For $k = 1$, recall that $\alpha_{h(0),h(1)} = 0$ by definition, and it trivially
holds that $R_{h(1)} - c_{h(1)} \leq R_{h(1)}.$ Now assume that the inequality holds for $k-1$, i.e.,
\begin{align}
R_{h(k-1)} - c_{h(k-1)} \leq \sum_{i = 1}^{k-1} (1-\alpha_{{h(i-1)},{h(i)}}) R_{h(i)}.
\label{eq:IH}
\end{align}
By construction we have $R_{h(k)} - c_{h(k)} \geq R_{h(k-1)}-c_{h(k-1)}$ and $R_{h(k)} > R_{h(k-1)}$, so we can apply Observation~\ref{obs:observation-6} to actions $a_{h(k)}$ and $a_{h(k-1)}$. This shows $(R_{h(k)} - c_{h(k)}) - (R_{h(k-1)} - c_{h(k-1)}) \leq (1 - \alpha_{h(k-1),h(k)})R_{h(k)}.$ Adding this to inequality (\ref{eq:IH}) we obtain
\begin{align*}
R_{h(k)} - c_{h(k)} \leq \sum_{i = 1}^{k} (1-\alpha_{{h(i-1)},{h(i)}}) R_{h(i)},
\end{align*}
as claimed.
\end{myproof}

The second crucial observation is that while $\alpha_{h(k-1),h(k)}$ is generally smaller than $\alpha_{h(k)}$ and thus does not incentivize action $h(k)$, it still $\delta$-incentivizes it.

\begin{claim}
\label{cla:delta-ic-approx-2}
For $k = 2, \dots, \kappa+1$, the linear contract with $\alpha = \alpha_{h(k-1),h(k)}$ ensures that $\alpha R_{h(k)} - c_{h(k)} +\delta \geq \alpha R_i - c_i$ for every $i \in [n]$.
\end{claim}
\begin{myproof}[Proof of Claim~\ref{cla:delta-ic-approx-2}]
The lines $R_{h(k)} - c_{h(k)}$ and $R_{h(k-1)} - c_{h(k-1)}$ intersect at $\alpha_{h(k-1),h(k)}$.
By construction, their intersection must fall between, on the one hand, the left endpoint $\gamma(1+\delta)^{k-2}$ of the interval in which $\alpha_{h(k)}$ falls, and $\alpha_{h(k)}$ on the other hand. 
This shows that $(1+\delta) \alpha_{h(k-1),h(k)} \geq (1+\delta) \gamma (1+\delta)^{k-2} = \gamma (1-\delta)^{k-1} \geq \alpha_{h(k)}$. Combining this with the fact that $a_{h(k)}$ is incentivized exactly at $\alpha_{h(k)}$, we obtain that $\alpha_{h(k-1),h(k)} R_{h(k)} - c_{h(k)} + \delta \ge (1+\delta)\alpha_{h(k-1),h(k)} R_{h(k)} - c_{h(k)} \geq \alpha_{h(k)} R_{h(k)} - c_{h(k)} \geq \alpha_{h(k)} R_i - c_i$ for all $i \in [n]$, where the first inequality holds since $R_{h(k)}\le 1$ by normalization. This completes the proof of Claim \ref{cla:delta-ic-approx-2}.
\end{myproof}

Using Claims \ref{cla:delta-ic-approx-1} and \ref{cla:delta-ic-approx-2}, the theorem follows from the fact that
\final{to obtain a $\delta$-IC linear contract we can} either incentivize $a_{h(1)}$ at $\alpha = \alpha_{h(1)}$ or $\delta$-incentivize one of the actions $a_{h(2)}, \dots, a_{h(\kappa+1)}$ at $\alpha = \alpha_{h(k-1),h(k)}$ with $2 \leq k \leq \kappa+1$.
Namely,
\begin{align*}
ALG 
&\geq \max\{(1-\alpha_{h(1)}) R_{h(1)}, (1-\alpha_{h(1),h(2)})R_{h(2)}, \dots, (1-\alpha_{h(\kappa),h(\kappa+1)})R_{h(\kappa+1)} \} \displaybreak[0]\\
&\geq (1-\gamma) \max\{(1-\alpha_{h(0),h(1)}) R_{h(1)}, (1-\alpha_{h(1),h(2)})R_{h(2)},\\
&\hspace*{201pt}\dots, (1-\alpha_{h(\kappa),h(\kappa+1)})R_{h(\kappa+1)} \} \displaybreak[0]\\
&\geq (1-\gamma) \frac{1}{\kappa+1} \sum_{i = 1}^{\kappa+1} (1-\alpha_{{h(k-1)},{h(k)}}) R_{h(k)} \displaybreak[0]\\
&\geq (1-\gamma) \frac{1}{\kappa+1} OPT,
\end{align*}
where for the first inequality applied to $ALG$ we use Claim~\ref{cla:delta-ic-approx-2}, for the second inequality we use $\alpha_{h(1)} \leq \gamma$ and $\alpha_{h(0),h(1)} \geq 0$, for the third inequality we lower bound the maximum by the average, and for the final inequality we use Claim~\ref{cla:delta-ic-approx-1}.
\end{proof}


\section{Black-box model}
\label{sec:black-box}

We conclude by considering a \emph{black-box model}
which concerns non-necessarily succinct principal-agent settings.
In this model, the principal knows the set of actions $A_n$, the cost $c_i$ of each action $a_i \in A_n$, the set of items $M$ and the rewards $r_j$ for each item $j \in M$, but does not know the probabilities $q_{i,S}$ that action $a_i$ assigns to outcome $S\subseteq M$. Instead, the principal has \emph{query access} to the distributions $\{q_i\}$. Upon querying distribution $q_i$ of action $a_i$, a (random) set is returned where $S$ is selected with probability $q_{i,S}$.
Our goal is to study how well a $\delta$-IC contract in this model can approximate the optimal IC contract if limited to a polynomial number of queries (where the guarantees should hold with high probability over the random samples).
Black-box models have been studied in other algorithmic game theory contexts such as signaling---see \cite{DughmiX16} for a successful example.

Let $\eta = \min\{q_{i,S} \mid i \in [n], S \subseteq M, q_{i,S} \neq 0\}$ be the minimum non-zero probability of any set of items under any of the actions. Note that then either $q_{i,S} = 0$ or $q_{i,S} \geq \eta$ for every $S$. In Section~\ref{sub:inverse-super} we address the case in which $\eta$ is inverse super-polynomial and obtain a negative result; in Section~\ref{sub:inverse-poly} we show a positive result for the case of inverse polynomial $\eta$.

\subsection{Inverse super-polynomial probabilities}
\label{sub:inverse-super} 

We show a negative result for the case where the minimum probability $\eta$ is inverse super-polynomial, by showing that $\textrm{poly}(1/\sqrt{\eta})$ samples are required to obtain a constant factor multiplicative approximation better than $\approx 1.15$. The negative result holds even for succinct settings, in which the unknown distributions are product distributions.

\begin{theorem}
Assume $\eta \leq \eta_0 = 1/625$ and $\delta \leq \delta_0 = 1/100$. Even with $n = 2$ actions and $m = 2$ items, achieving a multiplicative $\leq 1.15$ approximation to the optimal IC contract through a $\delta$-IC contract, where the approximation guarantee is required to hold with probability at least $1-\gamma$, may require at least $s \geq -\log(\gamma)/(9 \sqrt{\eta})$ queries.
\end{theorem}
\begin{proof}
We consider a scenario with two settings, both of which have $n = 2$ actions and $m = 2$ items, and which differ only in the probabilities of the items given the second action.
Let $\tau$ be some constant $> 2$ (to be fixed later), and
let $\mu=\frac{\sqrt{\eta}}{\tau}$.
Let $\beta = (1+\frac{1}{\tau^2})^{-1}$ and note that $\beta < 1$.

\medskip
\begin{center}
\begin{minipage}[t]{0.15\textwidth}
Setting I:
\end{minipage}
\begin{minipage}[t]{0.525\textwidth}
\begin{tabular}{@{}l|cc|l@{}}
\toprule
& $r_1 = \frac{\beta}{\tau^2 \mu}$ & $r_2 = \frac{\beta}{\tau^2 \mu}$ \\
\hline
$a_1:$ & $\tau \mu$ & $\tau \mu$ & $c_1 = 0$\\ 
$a_2:$ & $\tau^2 \mu$ & $\mu$ & $c_2 = \frac{\tau-1}{\tau^3} \frac{1}{1-\mu} \beta$\\ 
\bottomrule
\end{tabular}
\end{minipage}
\end{center}


\medskip

\begin{center}
\begin{minipage}[t]{0.15\textwidth}
Setting II:
\end{minipage}
\begin{minipage}[t]{0.525\textwidth}
\begin{tabular}{@{}l|cc|l@{}}
\toprule
& $r_1 = \frac{\beta}{\tau^2 \mu}$ & $r_2 = \frac{\beta}{\tau^2 \mu}$ \\
\hline
$a_1:$ & $\tau \mu$ & $\tau \mu$ & $c_1 = 0$\\
$a_2:$ & $\mu$ & $\tau^2 \mu$ & $c_2 = \frac{\tau-1}{\tau^3} \frac{1}{1-\mu} \beta$\\  
\bottomrule
\end{tabular}
\end{minipage}
\end{center}
\medskip

Note further that the minimum probability of any set of items in both settings is $q_{2,\{1,2\}}=\tau^2\mu^2=\eta$, as required by definition of $\eta$.

The expected reward achieved by the two actions in the two settings is $R_1 = 2\beta/\tau < 1$ and $R_2 = (1+1/\tau^2)\beta = 1$. Moreover, the cost of action $2$ is $c_2 \leq \beta/\tau^2$. So the welfare achieved by the two actions is $R_1 - c_1 < \beta$ and $R_2 - c_2 \geq \beta$.

In both settings the optimal IC contract incentivizes action $2$, by paying only for the set of items that maximizes the likelihood ratio. In Setting 1 this is $\{1\}$, in Setting 2 it is $\{2\}$. The payment for this set in both cases is $c_2/(\tau^2\mu(1-\mu) - \tau\mu(1-\tau\mu)) = c_2/(\tau^2\mu-\tau\mu)$. This leads to an expected payment of $\tau^2\mu(1-\mu) \cdot c_2/(\tau^2\mu-\tau\mu) = \beta/\tau^2$. The resulting payoff (and our benchmark) is therefore $R_2 - \beta/\tau^2 = \beta$.

We now argue that if we cannot distinguish between the two settings, then we can only achieve a $\approx 1.1568$ approximation. Of course, we can always pay nothing and incentivize action $1$, but this only yields a payoff of $2\beta/\tau$. We can also try to $\delta$-incentivize action $2$ in both settings, by paying for outcome $\{1\}$ and $\{2\}$. But (as we show below) the payoff that we can achieve this way is (for $\delta \rightarrow 0$ and $\mu \rightarrow 0$) at most $(1+1/\tau^2-(\tau^2+1)/((\tau-1)\tau^3)\beta$. Now $\max\{2/\tau,1+1/\tau^2-(\tau^2+1)/((\tau-1)\tau^3\}$ is minimized at $\tau = 1+\sqrt{2}$ where it is $2/(1+\sqrt{2}) \approx 0.8284$. The upper bound on the payoff from action $2$ for this choice of $\tau$ is actually increasing in both $\mu$ and $\delta$ and $\approx 0.8644 \cdot \beta$ at the upper bounds $\mu_0=\sqrt{\eta_0}/(2^2) = 1/100$ and $\delta_0 = 1/100$, implying that the best we can achieve without knowing the setting is a $\approx 1/0.8644 \approx 1.1568$ approximation.

So if we want to achieve at least a $\leq 1.15$ approximation with probability at least $1-\gamma$, then we need to be able to distinguish between the two settings with at least this probability. 
A necessary condition for being able to distinguish between the two settings is that we see at least some item in one of our queries to action $2$. So,
\[
1-\gamma \leq 1-(1-\tau^2\mu)^{2s},
\]
which implies that $s \geq \log(\gamma)/(2\log(1-\tau^2\mu) \geq -\log(\gamma)/(2 \cdot \mu \cdot\tau^2) \geq -\log(\gamma)/(18 \mu)$. Plugging in $\mu$ we get $s \ge -\log(\gamma)/(18 \frac{\sqrt{\mu}}{\tau}) > -\log(\gamma)/(9 \sqrt{\mu})$.

We still need to prove our claims regarding the payoff that we can achieve if we want to $\delta$-incentivize action $2$ in both settings.
To this end consider the IC constraints for $\delta$-incentivizing action 2 over action 1 in Setting I and Setting II, respectively:
\begin{align*}
&\tau^2\mu(1-\mu) p_{\{1\}} + (1-\tau^2\mu)\mu p_{\{2\}} - c_2 \geq \\ &\hspace*{50pt}\tau \mu(1-\tau\mu) p_{\{1\}} + (1-\tau\mu)\tau\mu p_{\{2\}} - \delta, \quad\text{and}
\\
&(1-\tau^2\mu)\mu p_{\{1\}} + \tau^2\mu(1-\mu)  p_{\{2\}} - c_2 \geq\\ &\hspace*{50pt}\tau \mu(1-\tau\mu) p_{\{1\}} + (1-\tau\mu)\tau\mu p_{\{2\}} - \delta.
\end{align*}
Adding up these constraints yields
\begin{align*}
&(\tau^2\mu(1-\mu) + (1-\tau^2\mu)\mu - 2 \tau \mu(1-\tau\mu)) \cdot (p_{\{1\}} + p_{\{2\}}) \geq 2 c_2 - 2 \delta.
\end{align*}
We maximize the minimum performance across the two settings by choosing $p_{\{1\}} = p_{\{2\}}$. Letting $p = p_{\{1\}} = p_{\{2\}}$ we thus obtain
\begin{align*}
(\tau^2\mu(1-\mu) + (1-\tau^2\mu)\mu - 2 \tau \mu(1-\tau\mu)) p &\geq c_2 - \delta.  
\end{align*}
It follows that
\[
p \geq \frac{c_2 - \delta}{\tau^2 \mu + \mu - 2 \tau \mu}.
\]
The performance of the optimal contract that $\delta$-incentivizes action 2 in both settings thus achieves an expected payoff of 
\begin{align*}
&R_2 - (\tau^2\mu(1-\mu)+(1-\tau^2\mu)\mu) \frac{c_2 - \delta}{\tau^2 \mu + \mu - 2 \tau \mu}= R_2 - \frac{\tau^2(1-2\mu)+1}{(\tau-1)^2} (c_2 - \delta).
\end{align*}
Plugging in $R_2$ and $c_2$ and letting $\delta \rightarrow 0$ and $\mu \rightarrow 0$ we obtain the aforementioned $1+1/\tau^2-(\tau^2+1)/((\tau-1)\tau^3)\beta$. Finally, to see that the expected payoff evaluated at $\tau = 1 + \sqrt{2} > 2$ is increasing in both $\delta$ and $\mu$ observe that the derivative in $\delta$ is simply the probability term $(\tau^2(1-2\mu)+1)/(\tau-1)^2$ which is positive and that both this probability term and the cost $c_2$ are decreasing in $\mu$ implying that as $\mu$ increases we subtract less.
\end{proof}

\subsection{Inverse polynomial probabilities}
\label{sub:inverse-poly}

We show a positive result for the case where the minimum probability $\eta$ is inverse polynomial. Namely, let $OPT$ denote the expected payoff of the optimal IC contract; then with $\textrm{poly}(n,m,\frac{1}{\eta},\frac{1}{\epsilon},\frac{1}{\gamma})$ queries it is possible to find with probability at least $(1-\gamma)$ a $4\epsilon$-IC contract with payoff at least $OPT - 5\epsilon$. Formally:

\begin{theorem}\label{thm:sample-positive}
Fix $\epsilon > 0$, and assume $\epsilon \leq 1/2$. Fix distributions $Q$ such that $q_{i,S} \geq \eta$ for all $i \in[n]$ and $S \subseteq M$.  Denote the expected payoff of the optimal IC contract for distributions $Q$ by $OPT$. Then there is an algorithm that with $s = (3 \log(\frac{2n}{\eta\gamma}))/(\eta \epsilon^2)$ queries to each action and probability at least $1-\gamma$, computes a contract $\tilde{p}$ which (i) is $4\epsilon$-IC on the actual distributions $Q$; and (ii) has expected payoff $\Pi$ on the actual distributions satisfying $\Pi \geq OPT - 5 \epsilon$.
\end{theorem}

To prove Theorem~\ref{thm:sample-positive}, we first prove a series of lemmas (Lemmas~\ref{lem:sample-bound} to~\ref{lem:noisy-2}). Proofs appear in Appendix~\ref{appx:black-box}.

\begin{lemma}
\label{lem:sample-bound}
Consider the algorithm that issues $s$ queries to each action $i \in N$, and sets $\tilde{q}_{i,S}$ to be the empirical probability of set $S$ under action $i$. With $s = (3 \log(\frac{2n}{\eta\gamma}))/(\eta \epsilon^2)$ queries to each action, with probability at least $1-\gamma$, for all $i \in [n]$ and $S \subseteq M$,
\begin{align*}
\quad &(1-\epsilon) q_{i,S} \leq \tilde{q}_{i,S} \leq (1+\epsilon) q_{i,S}.
\end{align*}
\end{lemma}

\begin{lemma}\label{lem:ic-is-preserved-approximately}
Suppose that $(1-\epsilon) q_{i,S} \leq \tilde{q}_{i,S} \leq (1+\epsilon) q_{i,S}$ for all $i \in [n]$ and $S \subseteq M$. Consider contract $p$. If $a_i$ is the action that is incentivized by this contract under the actual probabilities $Q$, then the payoff of $a_i$ under the empirical distributions $\tilde{Q}$ is at least as high as that of any other action up to an additive term of $2\epsilon$.
\end{lemma}

\begin{lemma}\label{lem:delta-ic-is-preserved}
Suppose that $(1-\epsilon) q_{i,S} \leq \tilde{q}_{i,S} \leq (1+\epsilon) q_{i,S}$ for all $i \in [n]$ and $S \subseteq M$. 
Consider contract $\tilde{p}$. If $a_i$ is the action that is $\delta$-incentivized by this contract under the empricial probabilities $\tilde{Q}$, then the payoff of $a_i$ under the actual distributions is at least as high as that of any other action up to an additive term of $\delta+2\epsilon$.
\end{lemma}

\begin{lemma}\label{lem:noisy-1}
Suppose that $(1-\epsilon) q_{i,S} \leq \tilde{q}_{i,S} \leq (1+\epsilon) q_{i,S}$ for all $i \in [n]$ and $S \subseteq M$. If action $a_i$ achieves payoff $\tilde\Pi$ under contract $\tilde{p}$ when evaluated on the empirical distributions $\tilde{Q}$, then it achieves payoff $\Pi \geq \tilde{\Pi} - 2\epsilon$ when evaluated on the actual distributions $Q$.
\end{lemma}

\begin{lemma}\label{lem:noisy-2}
Assume $\epsilon \leq 1/2$. Suppose that $(1-\epsilon) q_{i,S} \leq \tilde{q}_{i,S} \leq (1+\epsilon) q_{i,S}$ for all $i \in [n]$ and $S \subseteq M$.  If action $a_i$ achieves payoff $P$ under contract $p$ when evaluated on the actual distributions $Q$, then it achieves payoff $\tilde{P} \geq P - 3\epsilon$ when evaluated on the empirical distributions $Q$.
\end{lemma}

We are now ready to prove the theorem.

\begin{proof}[Proof of Theorem~\ref{thm:sample-positive}]
Compute empirical probabilities $\tilde{Q}$ by querying each action $s$ times. By Lemma~\ref{lem:sample-bound}, with probability at least $1-\gamma$, the empirical probabilities obtained in this way will satisfy
$(1-\epsilon) q_{i,S} \leq \tilde{q}_{i,S} \leq (1+\epsilon) q_{i,S}$ for all $i \in [n]$ and $S \subseteq M$. 

Suppose we compute the optimal $2\epsilon$-IC contract $\tilde{p}$ on the empirical distributions $\tilde{Q}$. Denote the expected payoff achieved by this contract on $\tilde{Q}$ by $\tilde{\Pi}$, and the expected payoff it achieves on $Q$ by~$\Pi$. 
Likewise, consider the optimal IC contract $p$ on the actual distributions $Q$. Denote the expected payoff $OPT$ achieved by this contract on the actual distributions $Q$ by $P$, and the expected payoff it achieves on $\tilde{Q}$ by $\tilde{P}$.

Note that by Lemma~\ref{lem:delta-ic-is-preserved}, contract $\tilde{p}$ which is $2\epsilon$-IC on $\tilde{Q}$ is $4\epsilon$-IC on $Q$. Furthermore, by Lemma~\ref{lem:ic-is-preserved-approximately}, contract $p$ which is IC on $Q$ is $2\epsilon$-IC on $\tilde{Q}$. This implies that $\tilde{\Pi} \geq \tilde{P}$. Together with Lemma~\ref{lem:noisy-1} and Lemma~\ref{lem:noisy-2} we obtain
\begin{align*}
\Pi \geq \tilde{\Pi} - 2\epsilon \geq \tilde{P} - 2\epsilon \geq P - 5\epsilon,
\end{align*}
which proves the theorem.
\end{proof}

\bibliographystyle{plainnat}
\bibliography{contracts_bib}

\appendix

\section{Basic properties of IC and $\mathbf{\delta}$-IC contracts} 
\label{appx:prelim}
In this appendix we state and prove several additional results concerning IC and $\delta$-IC contracts. 

\subsection{Intractability of the ellipsoid method}
\label{appx:ellipsoid}

We start by showing the intractability of the ellipsoid method for MIN-PAYMENT, except for the special case of $n=2$.
Recall LP~\eqref{LP:min-pay} for the MIN-PAYMENT problem. Its dual is as follows, where $\{\lambda_{i'}\}$ are $n-1$ nonnegative variables (one for every action other than $i$): 
\begin{align*}
\max~& \sum_{i'\ne i} {\lambda_{i'}(c_i-c_{i'})} \\
\text{s.t.}~& \big(\sum_{i'\ne i} {\lambda_{i'}}\big) - 1 \le \sum_{i'\ne i} {\lambda_{i'}\frac{q_{i',S}}{q_{i,S}}} &&\forall S\subseteq E,q_{i,S}>0, \\
& \lambda_{i'} \ge 0 &&\forall i'\ne i,i'\in [n].
\end{align*}

Consider applying the ellipsoid method to solve LP~\eqref{LP:min-pay} for action $a_i$. The separation oracle problem is: 
Given an instantiation of the dual variables $\{\lambda_{i'}\}$, consider the \emph{combination distribution} $\sum_{i'\ne i} \lambda_{i'} q_{i'}$, which is a convex combination of the product distributions $\{q_{i'}\}$. 
To find a violated constraint of the dual LP we need to find a set~$S$ for which the likelihood ratio between the combination distribution and the product distribution $q_i$ is sufficiently small.

Note that a combination distribution is \emph{not} itself a product distribution.%
\footnote{For example, consider a fifty-fifty mix between the following two product distributions over two items: a point mass on the empty set, and a point mass on the grand bundle. This combination distribution has probability $\frac{1}{2}$ for the empty set and probability $\frac{1}{2}$ for the grand bundle, and the item marginals are $\frac{1}{2}$. A product distribution with item marginals of $\frac{1}{2}$ has probability $\frac{1}{4}$ for every set.} 
Therefore solving the separation oracle is not easy and in fact it is an NP-hard problem even for $n=3$, as formalized in Proposition \ref{pro:separation-NP-hard}. In the special case of $n=2$, the combination distribution \emph{is} a product distribution. By taking $S$ to be all items that are more likely according to $q_i$ than according to the combination distribution, we minimize the likelihood ratio and solve the separation oracle. (This is one way to conclude that OPT-CONTRACT with $n=2$ is tractable.)

\begin{proposition}
	\label{pro:separation-NP-hard}
	Solving the separation oracle of dual LP~\eqref{LP:dual} is NP-hard for $n\ge 3$.
\end{proposition}

\begin{proof}
	Rather than prove Proposition~\ref{pro:separation-NP-hard} directly, it is enough to point the reader to Corollary~\ref{cor:implement-cost-hard}, which establishes the NP-hardness of MIN-PAYMENT. 
\end{proof}

\begin{remark}
	Proposition \ref{pro:separation-NP-hard} immediately holds for $\delta$-IC as well, i.e., for the separation oracle of dual LP~\eqref{LP:dual-tight}. This dual corresponds to primal LP~\eqref{LP:min-pay-relax} solving MIN-PAYMENT for $\delta$-IC contracts. This is simply because the separation oracle problem of dual LP~\eqref{LP:dual-tight} is identical to that of  dual LP~\eqref{LP:dual}.
\end{remark}

\subsection{Implementability and tractability of separable contracts}
\label{appx-sub:delta-IC-properties}

Next we state and prove two results concerning implementability and computability of $\delta$-IC contracts. Proposition~\ref{pro:delta-implementable} characterizes $\delta$-implementability and Proposition~\ref{pro:separable-poly-time} establishes tractability of $\delta$-IC separable contracts.

\begin{proposition}
	\label{pro:delta-implementable}
	For every $\delta>0$, every action $a_i$ can be $\delta$-implemented up to tie-breaking.
\end{proposition}

\begin{proof}
	Action $a_i$ can be $\delta$-implemented if and only if LP \ref{LP:delta-implement} has a feasible solution. 
	\begin{eqnarray}
	\min & 0 & \label{LP:delta-implement}\\
	\text{s.t.} & (1+\delta)\left(\sum_{S\subseteq E} {q_{i,S} p_S}\right) - c_i \ge \sum_{S\subseteq E} {q_{i',S} p_S} - c_{i'} & \forall i'\ne i,i'\in[n] \nonumber\\
	& p_S \ge 0 & \forall S\subseteq E.\nonumber
	\end{eqnarray}
	Consider the dual:
	\begin{eqnarray}
	\max & \sum_{i'\ne i} {\lambda_{i'}(c_i-c_{i'})} & \label{LP:dual-delta-implement}\\
	\text{s.t.} & (1+\delta)q_{i,S}\sum_{i'\ne i} {\lambda_{i'}} \le \sum_{i'\ne i} {\lambda_{i'}q_{i',S}} & \forall S\subseteq E,q_{i,S}>0 \nonumber\\
	& \lambda_{i'} \ge 0 & \forall i'\ne i,i'\in [n].\nonumber
	\end{eqnarray}
	Since $q_i$ and $\{q_{i'}\}$ are distributions and $\delta>0$, the only feasible solution to the dual LP~\eqref{LP:dual-delta-implement} is $\lambda_{i'}=0$ for every $i'\ne i$. The dual is feasible and bounded, hence the primal must be feasible, completing the proof.
\end{proof}

\begin{remark}
Proposition \ref{pro:delta-implementable} may seem surprising at first glance, but it is arguably less striking in comparison to implementability by IC contracts. Consider the computational problem IMPLEMENTABLE: The input is a succinct principal-agent setting and an action $a_i$, and the output is whether $a_i$ is implementable by an IC contract. The contract theory literature has characterized implementable actions as those whose distribution is \emph{not} a convex combination of other distributions with a lower combined cost (see,e.g.,~\cite{DRT18}). In general, for any set of product distributions corresponding to actions $a_{i'}\ne a_i$, only a trivial family of convex combinations preserve the product structure required in order to reconstruct the distribution of $a_i$; thus the answer to IMPLEMENTABLE is almost always YES.
\end{remark}

\begin{proposition}
	\label{pro:separable-poly-time}
	Let $\delta\ge 0$. Given a principal-agent setting, an optimal linear (resp., separable) $\delta$-IC contract can be found in polynomial time.
\end{proposition}

\begin{proof}
	The problem of finding an optimal linear (resp., separable) $\delta$-IC contract for incentivizing any action $a_i$ can be formulated as a polynomial-sized LP with 1 variable (resp., $m$ variables) representing the contract's parameter $\alpha$ (resp., the item payments $\{p_j\}$), and $n-1$ $\delta$-IC constraints. 
\end{proof}

\subsection{Connections between IC and $\delta$-IC contracts}
\label{appx-sub:IC-vs-delta}

We conclude this appendix with two results on the connection between IC and $\delta$-IC contracts. 
In Proposition \ref{pro:from-delta-to-IC} we show that from any contract that $\delta$-incentivizes action $a_i$, we can derive an IC contract with approximately the same expected payoff for the principal, up to a multiplicative factor of $(1-\sqrt{\delta})$ and an additive loss of $(\sqrt{\delta}-\delta)$. 
This is achieved by transferring to the agent a small fraction of the principal's expected payoff	from the original $\delta$-IC contract. Intuitively, such a transfer makes the agent's incentives ``more aligned'' with those of the principal, thus achieving incentive compatibility.  
In Proposition~\ref{pro:delta-IC-much-better} we show that the additive loss in Proposition~\ref{pro:from-delta-to-IC} is necessary: there can be a constant-factor gap between what a $\delta$-IC contract can achieve for the principal and what the optimal IC contract can achieve, even as $\delta\to 0$.  
Together, Propositions~\ref{pro:from-delta-to-IC} and~\ref{pro:delta-IC-much-better} paint an interesting and complete picture in comparison to auctions, where the relation between optimal $\epsilon$-IC auctions and optimal IC auctions is an open question. 

To state Proposition~\ref{pro:from-delta-to-IC},
denote by $\ell_{\alpha=1}$ the linear contract with parameter $\alpha=1$ (that transfers the full reward from principal to agent).

\begin{proposition}
	\label{pro:from-delta-to-IC} 
	Fix a principal-agent setting and $\delta>0$. Let $p$ be a contract that $\delta$-incentivizes action $a_i$. 
	Then the IC contract $p'$ defined as $(1-\sqrt{\delta})p + \sqrt{\delta}\ell_{\alpha=1}$ achieves for the principal expected payoff of at least $(1-\sqrt{\delta})(R_i-p_i) - (\sqrt{\delta}-\delta)$, where $R_i-p_i$ is the expected payoff of contract $p$.
\end{proposition}

\begin{proof}
	The expected payoff of action $a_i$ under the interpolation contract $p'$ is
	\[
	R_{i} -  [(1-\sqrt{\delta}) p_{i} + \sqrt{\delta} R_{i}] = (1-\sqrt{\delta}) (R_{i} -  p_{i}).
	\]
	We will argue that for every action $a_{i'}$ with $i' \neq i$, either $i'$ is not incentivized by $p'$ (Case 1) or its expected payoff is sufficiently high (Case 2).
	
	{\bf Case 1:} Assume $R_{i} - (1 + \sqrt{\delta}) p_{i} > R_{i'} - p_{i'}.$
	We claim that in this case $a_i$ is preferred over $a_{i'}$ under contract $p'$. Namely,
	\begin{align*}
	(1-\sqrt{\delta})p_{i} + \sqrt{\delta} R_{i} - c_{i} 
	&= (1+\delta) p_{i} - c_{i} + \sqrt{\delta}(R_{i} - (1+\sqrt{\delta})p_{i})\\
	&\geq p_{i'} - c_{i'}  + \sqrt{\delta}(R_{i} - (1+\sqrt{\delta})p_{i})\\
	&> p_{i'} - c_{i'}  + \sqrt{\delta}(R_{i'} - p_{i'})\\
	&= (1-\sqrt{\delta})p_{i'} + \sqrt{\delta} R_{i'} - c_{i'},
	\end{align*}
	where we used that action $a_i$ is $\delta$-incentivized under $p$ for the first inequality, and the second inequality holds by assumption because we are in Case 1.
	
	{\bf Case 2:} Assume now that $R_{i} - (1 + \sqrt{\delta}) p_{i} \leq R_{i'} - p_{i'}.$
	In this case the expected payoff achieved by action $a_{i'}$ is high. Namely,
	\begin{align*}
	R_{i'} - (1-\sqrt{\delta})p_{i'} - \sqrt{\delta} R_{i'} 
	&= (1-\sqrt{\delta}) (R_{a'_i} - p_{a'_i}) \\
	&\geq (1-\sqrt{\delta}) (R_{i} - (1 + \sqrt{\delta}) p_{i})\\
	&= (1-\sqrt{\delta}) (R_{i} - p_{i}) - (1-\sqrt{\delta}) \sqrt{\delta} p_{i},
	\end{align*}
	where the inequality holds by assumption because we are in Case 2.
\end{proof}

\begin{proposition}
	\label{pro:delta-IC-much-better}
	For any $\delta\in(0,\nicefrac{1}{2}]$, there exists a principal-agent setting where the optimal contract extracts expected payoff $OPT$ but a $\delta$-IC contract extracts expected payoff $\ge \frac{4}{3}OPT$ (and $OPT$ can be arbitrarily large).
\end{proposition}

\begin{proof}
	Consider the following principal-agent setting parameterized by $\delta$ and $\epsilon>0$. 
	Let $\mathcal{M}=\epsilon/\delta$. 
	There are $n=2$ actions and $m=2$ items. The probabilities of the items given the actions is described by the following matrix
	$$
	\begin{pmatrix}
	\frac{1}{4} & \frac{2\epsilon}{3(\mathcal{M}+\epsilon)} \\
	0 & 1
	\end{pmatrix},
	$$
	where the first column corresponds to item 1 and the second column to item 2.
	Set the rewards to be $r_1=\frac{4\epsilon}{3}$ for item 1 and $r_2=\mathcal{M}+\epsilon$ for item 2 (notice $r_1<r_2$), and the costs to be $c_1=0$ and $c_2=\mathcal{M}-\frac{\mathcal{M}\epsilon}{2(\mathcal{M}+\epsilon)}>0$. 
	Observe that the expected rewards are $R_1=\epsilon$ and $R_2=\mathcal{M}+\epsilon$. 
	
	\begin{claim}
		\label{cla:delta-better-1}
		$OPT=\epsilon$.
	\end{claim}
	
	\begin{myproof}[Proof of Claim \ref{cla:delta-better-1}]
		The expected payoff from letting the agent chose the zero-cost action $a_1$ is $R_1=\epsilon$. Can we get any better by incentivizing $a_2$?
		The optimal contract for incentivizing the costly action in a 2-action setting is well-understood (see e.g.~\cite{DRT18}): The only positive payment should be for the single subset of items maximizing the likelihood that the agent has chosen action $a_2$; in our case this is the subset $\{2\}$ containing item 2 only. Observe that its probability given action 1 is $\frac{\epsilon}{2(\mathcal{M}+\epsilon)}$. The 2-action characterization also specifies the payment for this outcome, setting it at 
		$p_{\{2\}} = c_2 / \left(1-\frac{\epsilon}{2(\mathcal{M}+\epsilon)}\right)=\mathcal{M}$. Subtracted from $R_2$ we get expected payoff of $\epsilon$ from optimally incentivizing $a_2$.
	\end{myproof}
	
	\begin{claim}
		\label{cla:delta-better-2}
		Contract $p$ that pays $\mathcal{M}-\frac{\epsilon}{3}$ for outcome $S=\{2\}$ and 0 otherwise $\delta$-incentivizes action $a_2$ with expected payoff $R_2-p_2=\frac{4}{3}\epsilon$. 
	\end{claim}
	
	\begin{myproof}[Proof of Claim \ref{cla:delta-better-2}]
		We show action $a_2$ is $\delta$-IC: The agent's expected utility from $a_1$ is $\frac{\epsilon}{2(\mathcal{M}+\epsilon)}p_2=\frac{\epsilon(3\mathcal{M}-\epsilon)}{6(\mathcal{M}+\epsilon)}$, and from $a_2$ given contract $(1+\delta)p$ it is $(1+\delta)p_2-c_2=(1+\frac{\epsilon}{\mathcal{M}})(\mathcal{M}-\frac{\epsilon}{3})-\mathcal{M}+\frac{\mathcal{M}\epsilon}{2(\mathcal{M}+\epsilon)}=\frac{\epsilon(2\mathcal{M}-\epsilon)}{3\mathcal{M}}+\frac{\mathcal{M}\epsilon}{2(\mathcal{M}+\epsilon)}$. 
		It can be verified that the former is less than the latter for $\delta \le \frac{1}{2}$. 
	\end{myproof}
	
	Putting these claims together completes the proof of Proposition \ref{pro:delta-IC-much-better}.
\end{proof}

\section{Hardness with a constant number of actions}
\label{appx:constant}
In this appendix we show NP-hardness of the two computational problems related to optimal contracts when the number of actions $n$ is constant. Appendices \ref{appx-sub:min-max-prob} and \ref{appx-sub:reduction} prove hardness of $\delta$-OPT-CONTRACT (Proposition~\ref{pro:implement-cost-hard}), from which hardness of $\delta$-MIN-PAYMENT follows by the reduction in Section \ref{sec:prelim} (Corollary~\ref{cor:implement-cost-hard}).

\begin{proposition}
	\label{pro:implement-cost-hard}
	\emph{$\delta$-OPT-CONTRACT} is NP-hard even for $n=3$ actions.
\end{proposition}  

\begin{corollary}
	\label{cor:implement-cost-hard}
	\emph{$\delta$-MIN-PAYMENT} is NP-hard even for $n=3$ actions.
\end{corollary}

\subsection{The computational problem MIN-MAX-PROB}
\label{appx-sub:min-max-prob}

It will be convenient to reduce to $\delta$-OPT-CONTRACT from a computational problem we call MIN-MAX-PROB, which is a variant of MIN-MAX PRODUCT PARTITION \cite{KP10} and thus NP-hard.
\begin{itemize}
	\item Input: A product distribution $q$ over $m$ items such that for every item $j$, its probability $q_j$ is equal to $\frac{1}{a_j+1}$ where $a_j$ is an integer $\in [3,a_{\max}]$ ($\log a_{\max}$ is polynomial in $m$). 
	\item Output: YES iff there exists a subset of items $S^*$ such that $q_{S^*}=\ell A$, where $A=\sqrt{\prod_j a_j}$ and $\ell=\prod_j q_j$. 
\end{itemize}
We now take a closer look at MIN-MAX-PROB. Denote $a_S=\prod_{j\in S}{a_j}$. 

\begin{observation}
	\label{obs:prob-vs-prod}
	The probability of subset $S$ is $q_S = \ell a_{\overline S}$. 
\end{observation}

\begin{proof}
	For every item $j$, the probability it is excluded is
	\begin{equation*}
	1-q_j=1-\frac{1}{a_j+1}=\frac{a_j}{a_j+1}=q_ja_j.
	\end{equation*}
	So the probability of the outcome being precisely $S$ is
	\begin{align*}
	q_{S} 
	&= \left(\prod_{j\in S} q_j \right)\left(\prod_{j\notin S} (1-q_j)\right) \\
	&= \left(\prod_{j\in S} q_j\right) \left(\prod_{j\notin S} q_ja_j\right)\\
	&= \left(\prod_{j=1}^m q_j \right) \left(\prod_{j\notin S} a_j \right)
	= \ell a_{\overline S},
	\end{align*}
as claimed.
\end{proof}

Observation \ref{obs:prob-vs-prod} immediately implies:

\begin{observation}
	\label{obs:min-max}
	For every subset $S$, $a_S+a_{\overline S}=a_S+\frac{A^2}{a_S}\ge 2A$,
	where equality holds iff $a_S=a_{\overline S}=A$. Equivalently, $q_S+q_{\overline S}\ge 2\ell A$, where equality holds iff $q_S=q_{\overline S}=\ell A$.
\end{observation}
\begin{proof}
	The inequality in the observation holds by the inequality of arithmetic and geometric means (AM-GM inequality), which states that for any two non-negative numbers $w, z$, $(w+z)/2 \geq \sqrt{wz}$. Namely, for $z = a_S$, $w = A^2/a_S$, and $A = \sqrt{zw}$ the AM-GM inequality states that $a_S+A^2/a_S = z + w \geq 2 \sqrt{wz} = 2 \sqrt{a_S \cdot A^2/a_S} = 2 A$ as claimed.
\end{proof}

Observation \ref{obs:min-max} shows the connection between MIN-MAX-PROB and the NP-hard problem MIN-MAX PRODUCT PARTITION: $q$ is a YES instance (there exists a subset of items $S$ such that $q_{S}=\ell A$) iff $a_S=A$.

The following observation will be useful in the reduction to $\delta$-OPT-CONTRACT.

\begin{observation}
	\label{obs:delta}
	Let $\Delta = 1-\ell A 2^{m-1}$, then $0 < \Delta < 1$. 
\end{observation}

\begin{proof}
	By definition,
	$$
	\ell A = \frac{\sqrt{\prod a_j}}{\prod (a_j+1)} 
	\le \frac{\prod \sqrt{a_j +1}}{\prod (a_j +1)}=\frac{1}{\prod \sqrt{a_j+1}}\le \frac{1}{2^{m}}< \frac{1}{2^{m-1}},
	$$
	where the second-to-last inequality follows since $a_j\ge 3$ and so $\sqrt{a_j+1}\ge 2$.
	We conclude that $\ell A 2^{m-1} < 1$, completing the proof. 
\end{proof}


\subsection{Proof of Proposition~\ref{pro:implement-cost-hard}}
\label{appx-sub:reduction}

We now use hardness of MIN-MAX-PROB to establish hardness of $\delta$-OPT-CONTRACT.

\begin{proof}[Proof of Proposition \ref{pro:implement-cost-hard}]
	The proof is by reduction from MIN-MAX-PROB, as follows.
	
	{\bf Reduction.} Given an instance $q$ of MIN-MAX-PROB, construct a principal-agent setting with $n=3$ actions. 
	\begin{itemize}
		\item For action $a_1$, set its product distribution $q_1$ to be $q$. 
		\item For action $a_2$, set its product distribution $q_2$ to be $1-q$ (i.e., $q_{1,j}+q_{2,j}=1$ for every item $j$). 
		\item For action $a_3$, set its product distribution $q_3$ to be such that $q_{3,1}=1$ (i.e., this action's outcome always includes item $1$), and $q_{3,j}=\frac{1}{2}$ for every other item $j>1$. 
	\end{itemize}
	Set costs $c_1,c_2$ to zero and set $c_3$ to be $c=(a_{\max}+1)^{-1}$. The only nonzero reward is $r=r_1$ for item 1; set $r$ to be any number strictly greater than $\Delta^{-1}$. 
	
	{\bf Analysis.} 
	First notice that the reduction is polynomial in $m$; in particular, the number of bits of precision required to describe the probabilities, cost $c$ and reward $r$ is polynomial.
	
	The analysis will show that the expected payoff the principal can extract by a $\delta$-IC contract if $q$ is a YES instance is strictly larger than if $q$ is a NO instance. 
	We introduce some notation: Let $\mathcal{S}^1=\{S\subseteq [m]\mid 1\in S\}$, i.e., $\mathcal{S}^1$ is the collection of all item subsets containing item~1. Given a contract $p$, let $P = \sum_{S\in \mathcal{S}^1} p_S$ (the total payment for subsets in $\mathcal{S}^1$). Observe that the expected payment to the agent if he chooses action $a_3$ is $\frac{P}{2^{m-1}}$. 
	
	\begin{claim}
		\label{cla:a3-min-pay}
		Action $a_3$ can be weakly $\delta$-incentivized with expected payment $\frac{c}{\Delta(1+\delta)}$ if and only if $q$ is a YES instance of MIN-MAX-PROB. 
	\end{claim}
	
	\begin{myproof}[Proof of Claim \ref{cla:a3-min-pay}]
		Fix a $\delta$-IC contract $p$ that weakly $\delta$-incentivizes action $a_3$. By Observation~\ref{obs:prob-vs-prod}, the agent's expected utility from action $a_1$ is $\ell\sum_S p_S a_{\overline S}$ and from action $a_2$ is $\ell\sum_S p_S a_{S}$. The agent's expected utility from action $a_3$  (after boosting by $(1+\delta)$) is $\frac{P(1+\delta)}{2^{m-1}} - c$. 
		
		Assume first that $q$ is a NO instance. If $p$ weakly incentivizes action $a_3$ then
		\begin{eqnarray*}
			\frac{P(1+\delta)}{2^{m-1}} - c &\ge& \ell\cdot \max\left\{\sum_S p_S a_{S},\sum_S p_S a_{\overline S}\right\}\nonumber\\ 
			&\ge& \frac{\ell}{2}\left( \sum_S p_S a_{S} + \sum_S p_S a_{\overline S} \right)\nonumber\\
			&=& \frac{\ell}{2} \sum_S p_S (a_{S} +a_{\overline S}) > \ell A\sum_S p_S \ge \ell AP,
		\end{eqnarray*}
		where the second-to-last inequality is by Observation \ref{obs:min-max}, and is strict by our assumption that $q$ is a NO instance. Rearranging $\frac{P(1+\delta)}{2^{m-1}} - c > \ell A P$ we get 
		\begin{equation*}
		c < \frac{P(1+\delta)}{2^{m-1}} - \ell A P(1+\delta) = \frac{P(1+\delta)}{2^{m-1}}\left(1 - \ell A 2^{m-1} \right) = \frac{P\Delta(1+\delta)}{2^{m-1}}.
		\end{equation*}
		By Observation \ref{obs:delta} we can divide both sides by $\Delta(1+\delta) > 0$ to establish $\frac{P}{2^{m-1}} > \frac{c}{\Delta(1+\delta)}$, completing the proof of the first direction.
		
		Assume now that $q$ is a YES instance. Then there exists $S^*$ such that $a_{S^*}=a_{\overline {S^*}}=A$, and without loss of generality $S^*\in \mathcal{S}^1$ (otherwise take its complement). Consider the following contract: Let $p_{S^*}=\frac{c2^{m-1}}{\Delta(1+\delta)}$ and set all other payments to 0. 
		The expected payment to the agent for action $a_3$ is	$\frac{p_{S^*}}{2^{m-1}}=\frac{c}{\Delta(1+\delta)}$ as required, and the agent's expected utility (after boosting by $(1+\delta)$) is $\frac{p_{S^*}(1+\delta)}{2^{m-1}}-c=\frac{c}{\Delta}-c =\frac{c(1-\Delta)}{\Delta}$. Plugging in $\Delta= 1-\ell A2^{m-1}$, we get that the expected utility from action $a_3$ is $\ell A\frac{c2^{m-1}}{\Delta}=\ell A p_{S^*}$. This is equal to the expected utility from action $a_1$, since $\ell\sum_S p_S a_{\overline S}=\ell p_{S^*} a_{\overline {S^*}}=\ell A p_{S^*}$ 
		Similarly, the expected utility from action $a_2$ is also $\ell A p_{S^*}$. We conclude that $p$ weakly $\delta$-incentivizes $a_3$,
		completing the proof of Claim \ref{cla:a3-min-pay}.
	\end{myproof}
	%
	
	
	We now use Claim \ref{cla:a3-min-pay} to complete the hardness proof by showing that the expected payoff the principal can extract if $q$ is a YES instance is strictly larger than if $q$ is a NO instance. 
	
	For a YES instance, by Claim \ref{cla:a3-min-pay} action $a_3$ can be weakly $\delta$-incentivized with expected payment $\frac{c}{\Delta(1+\delta)}$. 
	We argue that the values chosen in the reduction for $c$ and $r$ guarantee that action $a_3$ has the (strictly) highest expected payoff for the principal, so the agent breaks ties in favor of~$a_3$: Since the only positive reward is $r_1=r$ and since $q_{3,1}=1$, the expected payoff from $a_3$ is $q_{3,1}r_1-\frac{c}{\Delta(1+\delta)}=r-\frac{c}{\Delta(1+\delta)}$. The expected reward (and thus also payoff) from $a_1$ is at most $q_{1,1}r_1\le \frac{r}{4}$ (using that $a_1+1\ge 4$), and the expected reward from $a_2$ is at most $q_{2,1}r_1\le (1-\frac{1}{a_{\max}+1})r$. Since $\frac{r}{4} \le (1-\frac{1}{a_{\max}+1})r$ (using that $a_{\max}\ge 3$), it suffices to show $r-\frac{c}{\Delta(1+\delta)}\ge r-\frac{c}{\Delta}> (1-\frac{1}{a_{\max}+1})r$, or simplifying, $r>\frac{c(a_{\max}+1)}{\Delta}$. Since the reduction sets $c=(a_{\max}+1)^{-1}$ and $r >\Delta^{-1}$, the argument is complete.
	
	For a NO instance, by Claim \ref{cla:a3-min-pay} the expected payoff from $a_3$ is strictly lower than $r-\frac{c}{\Delta(1+\delta)}$. By the analysis of the YES case we know that the expected rewards from $a_1,a_2$ are strictly lower than $r-\frac{c}{\Delta}$ (and by limited liability the principal's expected payoff is bounded by the expected reward). This completes the proof of Proposition \ref{pro:implement-cost-hard}.
\end{proof}

\section{An FPTAS for the separation oracle}
\label{appx:separation-oracle}
 In this appendix we establish the separation oracle FPTAS stated in Lemma~\ref{lem:FPTAS}. 



\begin{proof}[Proof of Lemma \ref{lem:FPTAS}]
	We adapt an FPTAS of Moran~\cite{Moran81} (see also subsequent papers such as \cite{NBCK10}).
	Let 
	$$
	\Delta=(1+\epsilon)^{1/2m}.
	$$
	
	{\bf FPTAS algorithm.} 
	The algorithm proceeds in iterations from $0$ to $m$. 
	In iteration $j$, the partial solutions in that iteration are subsets of the first $j$ items. 
	For a partial solution $S\subseteq \{1,\dots,j\}$, recall that $q_{\ell,S}$ is the marginal probability to draw $S$ among the first $k$ items if the sample is distributed according to $q_\ell$.
	
	The partial solutions in iteration $j$ are partitioned into families $Y_{j,1},\dots,Y_{j,r_j}$. The partition is such that for every family $r\in[r_j]$ and partial solutions $S,S'\in Y_{j,r}$, for every distribution $\ell\in [k]\cup \{i\}$, the ratio between $q_{\ell,S}$ and $q_{\ell,S'}$ is at most $\Delta$.
	
	In the first iteration $j=0$, the only solution is the empty set. 
	The solutions in iteration $j+1$ are generated from the families in iteration $j$ as follows: One arbitrary partial solution $S$ is chosen from every family $Y_{j,r}$ to ``represent'' it, and for each such $S$ two partial solutions $S\cup \{j+1\}$ and $S$ are added to the solutions of iteration $j+1$ (i.e., with and without the $(j+1)$st item).
	
	The algorithm outputs the minimum objective $\frac{1}{q_{i,S}} \sum_{k} {\alpha_{k} q_{k,S}}$ among the solutions $S$ in iteration~$m$.
	
	{\bf Analysis.}
	We first argue that $ALG\le (1+\epsilon)OPT^s$. Let $S^*$ be the optimal solution, 
	and denote the subset of 
	$S^*$ containing only items among the first $j$ by 
	$S^*_j$. 
	By induction, in iteration $j$ there is a partial solution $S'_j$ such that $\Delta^{-j}\cdot q_{\ell,S^*_j} \le q_{\ell,S'_j}\le \Delta^j \cdot q_{\ell,S^*_j}$ for every distribution $\ell\in [k]\cup \{i\}$. Denote $S'=S'_m$. Then $ALG\le 
	\frac{1}{q_{i,S'}} \sum_{k} {\alpha_{k} q_{k,S'}} \le 
	\Delta^{2m}\cdot \frac{1}{q_{i,S^*}} \sum_{k} {\alpha_{k} q_{k,S^*}}=(1+\epsilon)OPT^s$.
	
	
	
	It remains to show that the FPTAS runs in polynomial time. The running time is $O(\sum_j r_j)$.
	In the input distributions $\{q_k\},q_i$, denote the range of every \emph{nonzero} probability by $[q_{\min},1]$ ($q_{\min}$ can be exponentially small). For every distribution $\ell\in [k]\cup \{i\}$, the probabilities that are not $0$ are at least $q_{\min}^m$. So a partition ``in jumps of $\Delta$'' requires $O(t)$ parts, where $t$ is the smallest integer satisfying $q_{\min}^m\cdot \Delta^ t \ge 1$. So 
	$$
	t = \left\lceil \frac{m \log (q_{\min}^{-1})}{\log {\Delta}} \right\rceil = \left\lceil \frac{2m^2 \log (q_{\min}^{-1})}{\log {(1+\epsilon)}} \right\rceil
	\le \left\lceil \frac{2m^2 \log (q_{\min}^{-1})}{\epsilon} \right\rceil,
	$$ 
	where the last inequality uses $\log (1+\epsilon)\ge \epsilon$ for $\epsilon\in(0,1]$. Since the partition to $r_j$ families maintains ``jumps of $\Delta$'' for $n$ distributions, $r_k=O(t^n)$. We invoke the assumption that $n$ is constant to complete the analysis and the proof of Lemma \ref{lem:FPTAS}.
\end{proof}

\section{Hardness of MIN-PAYMENT}
\label{appx:hardness-min-payment}
In this appendix we show the following counterpart to Corollary~\ref{cor:hardness-of-approx}. 

\begin{proposition}
	\label{pro:min-pay-hardness}
	For any constant $c\in \mathbb{R},c\ge 1$, 
	it is NP-hard to approximate the minimum expected payment for 
	implementing a given action to within a multiplicative factor $c$. 
\end{proposition}

\begin{proof}
The proof is by reduction from MAX-3SAT. Given an instance of MAX-3SAT, the goal is to determine whether the instance is satisfiable or whether at most $\frac{7}{8}+\epsilon$ of the clauses can be satisfied, where $\epsilon$ is an arbitrarily small constant. 

{\bf Reduction.} Given $\varphi$, we obtain the SAT principal-agent setting corresponding to $\varphi$ (Proposition \ref{pro:sat-setting}), but we set the reward for every item to be $1$ rather than $0$. We add an action $a_{n+1}$ with cost $\mathcal{C}$ and product distribution $q_{n+1}$ with probability $\frac{1}{2}$ for every item. 

{\bf Analysis.} As in the analysis in the proof of Proposition \ref{pro:hardness-reduction-2}, if $\varphi$ has a satisfying assignment then we can implement $a_{n+1}$ at cost $\mathcal{C}$. 
Otherwise recall that by Definition \ref{def:sat-setting}, the average action over the first $n$ actions leads to every item set $S$ with probability at least $\frac{1-8\epsilon}{2^m}$. 
Consider a contract $p$ and let $P=\sum_S p_S$. 
The expected utility of the agent for choosing $a_{n+1}$ is $\nicefrac{P}{2^m} - \mathcal{C}$. Consider again the average action over the first $n$ actions.
The expected payment to the agent for ``choosing'' this action (i.e., the expected payment over the average distribution) is at least $\frac{1-8\epsilon}{2^m}P=\frac{P}{2^m} - \frac{8\epsilon P}{2^m}$, and
there is some action $a_i$ (with cost $0$) for which the expected payment is as high. 
To incentivize $a_{n+1}$ over $a_i$ it must hold that $\frac{P}{2^m} - \mathcal{C} \ge \frac{P}{2^m} - \frac{8\epsilon P}{2^m}$, i.e., $\frac{P}{2^m} \ge \frac{\mathcal{C}}{8\epsilon}$.
We conclude that if there is no assignment satisfying more than $\frac{7}{8}+\epsilon$ of the clauses, the expected payment for implementing $a_{n+1}$ is $\frac{\mathcal{C}}{8\epsilon}$ rather than~$\mathcal{C}$. Approximating the expected payment within a multiplicative factor $\frac{1}{8\epsilon}$ would thus solve the MAX-3SAT instance we started with, and we can make $\epsilon$ as small a constant as we want.
\end{proof}

\section{Proofs omitted from Section~\ref{sec:hardness}} 
\label{appx:hardness}
In this appendix we provide proofs for Propositions~\ref{pro:gap-setting-two-actions},~\ref{pro:gap-setting}, and~\ref{pro:sat-setting}.

\subsection{Existence of gap settings}
\label{appx:gap-settings}

We start by establishing the existence of gap settings for $2$ actions (Proposition~\ref{pro:gap-setting-two-actions}) and $c$ actions (Proposition~\ref{pro:gap-setting}).

\begin{proof}[Proof of Proposition \ref{pro:gap-setting-two-actions}]
	For the gap setting constructed above with $c=2$ actions and $\gamma=\epsilon$, consider a $\delta$-IC contract. Since the expected reward of the first action $a_1$ is 1, and the maximum expected welfare is $2-\gamma\ge 2-\frac{4\epsilon}{1+2\epsilon}$, if a contract is to extract more than $\frac{1}{2-4\epsilon/(1+2\epsilon)}=\frac{1}{2}+\epsilon$ of the expected welfare then it must $\delta$-incentivize the last action $a_c$ (a limited liability contract cannot extract more than the expected reward from an agent choosing $a_1$, since $a_1$ is zero-cost). Let $p$ be the payment for the item and let $p_0$ be the payment for the empty set. 
	For any action $a_{i^*}$ that the contract $\delta$-incentivizes, 
	the following inequality must hold for every~$i\in [c]$:  
	\begin{eqnarray}
	(1+\delta)\left(\gamma^{c-i^*} p + (1-\gamma^{c-i^*}) p_0\right) &-& \frac{1}{\gamma^{i^*-1}} + i^* - (i^*-1)\gamma ~\ge\nonumber\\ 
	\left(\gamma^{c-i} p + (1-\gamma^{c-i}) p_0\right) &-& \frac{1}{\gamma^{i-1}} + i - (i-1)\gamma.\label{eq:IC-gap}
	\end{eqnarray}
	Observe that for the contract to $\delta$-incentivize $a_c$ at minimum expected payment, it must hold that $p_0=0$. 
	We can now plug $p_0=0$ into inequality \eqref{eq:IC-gap} and choose $i^*=c,i=i^*-1$. We get a lower bound on the expected payment for $\delta$-incentivizing $a_c$:
	$$
	p\ge \frac{(1-\gamma)^2}{\gamma(1+\delta-\gamma)}.
	$$ 
	The principal's expected payoff is thus $\le \frac{1}{\gamma} - \frac{(1-\gamma)^2}{\gamma(1+\delta-\gamma)}\le \frac{1}{1+\gamma^2-\gamma}$, where the last inequality uses $\delta\le f(\epsilon) = \gamma^2$.	
	We get an upper bound of $\frac{1}{1+\gamma^2-\gamma}$ on what the best $\delta$-IC contract can extract out of $2-\gamma$ for the principal. The ratio is thus at most $\frac{1}{2}+\epsilon$ (using $\gamma\le \frac{1}{4}$), and this completes the proof of Proposition~\ref{pro:gap-setting-two-actions}.
\end{proof}

\begin{proof}[Proof of Proposition \ref{pro:gap-setting}]
	For the gap setting constructed above with $c$ actions and $\gamma=\epsilon$, consider a $\delta$-IC contract. 
	As in the proof of Proposition \ref{pro:gap-setting-two-actions}, this contract cannot extract more than $\frac{1}{c}+\epsilon$ of the expected welfare by $\delta$-incentivizing action $a_1$.
	Assume from now on that the contract $\delta$-incentivizes action $a_{i^*}$ for $i^*\ge 2$ at minimum expected payment. 
	As in the proof of Proposition \ref{pro:gap-setting-two-actions}, Inequality~\eqref{eq:IC-gap} must hold for $i^*$ and every $i\in[c]$.
	
	Assume first that the contract's payment $p_0$ for the empty set is zero. (This assumption is without loss of generality for the case of $c=2$ actions, as well as for $c\ge 3$ and fully-IC optimal contracts by Proposition 6 in \cite{DRT18}.) Plugging $p_0=0$ into Inequality~\eqref{eq:IC-gap} and choosing $i=i^*-1$, we get a lower bound on the expected payment for $\delta$-incentivizing $a_{i^*}$ (in particular making it preferable to $a_{i^*-1}$): 
	\begin{equation}
	\gamma^{c-i^*}p \ge \frac{(1-\gamma^{i^*-1})(1-\gamma)}{\gamma^{i^*-1}(1+\delta-\gamma)}.\label{eq:bound-for-zero-p0}
	\end{equation}
	The principal's expected payoff is thus $\le \frac{1}{\gamma^{i^*-1}} - \frac{(1-\gamma^{i^*-1})(1-\gamma)}{\gamma^{i^*-1}(1+\delta-\gamma)} \le \frac{\gamma^c+\gamma^{i^*-1}(1-\gamma)}{\gamma^{i^*-1}(1+\gamma^{c}-\gamma)}
	= \frac{\gamma^c}{\gamma^{i^*-1}(1+\gamma^{c}-\gamma)} + \frac{1-\gamma}{1+\gamma^{c}-\gamma}$, 
	where the last inequality uses $\delta \le f(\epsilon) = \gamma^c$.
	Maximizing this expression by plugging in $i^*=c$, we get an upper bound of $\frac{1}{1+\gamma^{c}-\gamma}$ on what the best $\delta$-IC contract can extract out of $c-(c-1)\gamma$ for the principal. 
	The ratio can thus be shown to be at most $\frac{1}{c} + \epsilon$, as required (using that $c\ge 3$ and $\gamma\le \frac{1}{4}$; see Claim \ref{cla:calc}). 
	
	Now consider the case that $p_0>0$. We argue that in this case, plugging $i=i^*-1$ into Inequality~\eqref{eq:IC-gap} gives a lower-bound on $\gamma^{c-i^*}p$ that is only higher than that in Inequality~\eqref{eq:bound-for-zero-p0}. To see this, consider the contribution of $p_0>0$ to the left-hand side of Inequality~\eqref{eq:IC-gap}, which is $(1+\delta)(1-\gamma^{c-i^*})p_0$. Compare this to its contribution to the right-hand side of Inequality~\eqref{eq:IC-gap}, which is $(1-\gamma^{c-i})p_0$. For $\delta\le \gamma^c$, $\gamma\le \frac{1}{4}$ and $i=i^*-1$ it holds that $(1+\delta)(1-\gamma^{c-i^*}) \le 1-\gamma^{c-i}$. 
	This completes the proof of Proposition \ref{pro:gap-setting} up to Claim~\ref{cla:calc}.
\end{proof}

\begin{claim}
For every $\gamma\in(0,\frac{1}{4}]$ and $c\in\mathbb{Z},c\ge 3$,
\label{cla:calc}
$$
\frac{1}{1+\gamma^{c}-\gamma} \cdot \frac{1}{c-(c-1)\gamma}\le \frac{1}{c} + \gamma.
$$
\end{claim}

\begin{proof}
	We first establish the claim for $c=3$. We need to show $\frac{1}{1+\gamma^{3}-\gamma} \cdot \frac{1}{3-2\gamma}\le \frac{1}{3} + \gamma$. Simplifying, we need to show $13\gamma+6\gamma^4\le 4+9\gamma^2+7\gamma^3$, which holds for every $\gamma\le\frac{1}{4}$. 
	
	We now consider $c\ge 4$: It is sufficient to show $\frac{1}{1-\gamma} \cdot \frac{1}{c-c\gamma}\le \frac{1}{c} + \gamma$. Multiplying by $c$ we get $\frac{1}{(1-\gamma)^2} \le 1 + c\gamma$. This holds if and only if $c\ge \frac{2-\gamma}{(1-\gamma)^2}$. The right-hand side is an increasing function in the range $0<\gamma \le \frac{1}{4}$ and so we can plug in $\gamma=\frac{1}{4}$ and verify. Since $c\ge 4\ge \frac{28}{9}$, the proof is complete.
\end{proof}

\subsection{Proof of Proposition~\ref{pro:sat-setting}}
\label{appx:sat-setting}
Next we show that Algorithm~\ref{alg:sat-setting-from-phi} provides an efficient reduction from MAX-3SAT instances to gap settings (Proposition~\ref{pro:sat-setting}).

\begin{proof}[Proof of Proposition \ref{pro:sat-setting}]
	We first argue that there is a satisfying assignment to the MAX-3SAT instance if and only if there is a set $S$ with $0$-probability in every one of the product distributions.
	First note that there is a natural 1-to-1 correspondence between subsets $\{S\}$ of items and truth assignments to the variables: for every variable $j$, if item $j\in S$ then assign TRUE and otherwise FALSE.
	Now consider a set $S$ and its corresponding assignment. $S$ has $0$-probability in the $i$th product distribution iff either an item in $S$ has probability 0 or an item in $\overline{S}$ has probability 1 according to this distribution. Therefore,
	in clause $i$, either one of the TRUE variables appears as a positive literal or one of the FALSE variables appears as a negative literal. And this is a necessary and sufficient condition for the clause to be satisfied. We conclude that $S$ has $0$-probability in every product distribution if and only if the corresponding assignment satisfies every clause, establishing condition (1) of Definition~\ref{def:sat-setting}.
	To show condition (2), assume that at most $\frac{7}{8}+\alpha$ of the clauses can be satisfied. Consider the average action whose distribution results from averaging over all actions. This distribution has for every $S$ a probability at least $(\frac{1}{8}-\alpha)\cdot\frac{8}{2^m}=\frac{1-8\alpha}{2^m}$, since the probability of $S$ is $\frac{8}{2^m}$ in every distribution corresponding to a clause which the assignment corresponding to $S$ does not satisfy. This completes the proof.
\end{proof}

\section{Approximation by separable contracts} 
\label{appx:separable}
In this appendix we examine the gap between separable and optimal contracts.

Recall that a contract $p$ is \emph{separable} if there are payments $p_1,...,p_m$ such that $p(S) = \sum_{j\in S} p_j$ for every $S \subseteq M$.
By linearity of expectation, the expected payment for action $a_i$ given a separable contract $p$ is $\sum_{j} q_{i,j}p_j$.

As we have shown in Proposition~\ref{pro:separable-poly-time} the optimal separable contract can be computed in polynomial time via linear programming.
Thus we know that separable (and other simple computationally-tractable) contracts cannot achieve a constant approximation to OPT unless $P=NP$ (Corollary~\ref{cor:hardness-of-approx}). 

In fact, an even stronger lower bound holds---they cannot achieve an approximation better than $n$, unless we relax the IC requirement to $\delta$-IC. We provide a proof of this general lower bound for the case of $n=2$.

\begin{proposition}
	For every $\epsilon > 0$ there is a principal-agent instance with $n=2$ actions and $m = 2$ items, in which the best separable contract only provides a $2-\epsilon$ approximation to $OPT$.
\end{proposition}

\begin{proof}
	For $\delta \in (0,1)$ consider the following $n=2$ actions and $m = 2$ items instance. The probabilities $q_{i,j}$ for the two actions $i \in \{1,2\}$ and items $j \in \{1,2\}$ are 
	\begin{align*}
	&q_{1,1} = \frac{\delta}{2}, \quad q_{1,2} = 1 - \frac{\delta}{2} \quad \text{and} \quad q_{2,1} = \frac{1}{2}, \quad q_{2,2} = \frac{1}{2}.
	\end{align*}
	The rewards $r_j$ for the two items $j \in \{1,2\}$ are
	\begin{align*}
	r_1 = \frac{1-(1-\frac{\delta}{2})\delta}{\frac{\delta}{2}} \quad \text{and} \quad r_2 = \delta.
	\end{align*}
	The resulting expected rewards $R_i$ for the two actions $i \in \{1,2\}$ are
	\begin{align*}
	R_1 &= q_{1,1} r_1 + q_{1,2} r_2 = \frac{\delta}{2} \frac{1-(1-\frac{\delta}{2})\delta}{\frac{\delta}{2}} + (1-\frac{\delta}{2})\delta = 1, \quad \text{and} \\
	R_2 &= q_{2,1} r_1 + q_{2,2} r_2 = \frac{1}{2} \frac{1-(1-\frac{\delta}{2})\delta}{\frac{\delta}{2}}+\frac{1}{2}\delta = \frac{1}{\delta}-1+\delta,
	\end{align*}
	so that $R_2 > 1$ for all $\delta \in (0,1)$ and $R_2 \rightarrow \infty$ as $\delta \rightarrow 0$.
	The costs $c_i$ for the two actions $i \in \{1,2\}$ are
	\begin{align*}
	c_1 = 0 \quad \text{and} \quad c_2 = (1-\delta) (R_2 - R_1) = (1-\delta) (\frac{1}{\delta} - 2 + \delta).
	\end{align*}
	Note that on this instance 
	\begin{align*}
	R_1 - c_1 = 1 \quad \text{and} \quad R_2-c_2 = 2- 2\delta+\delta^2.
	\end{align*}
	
	We claim that: (1) The optimal contract can incentivize action 2 with an expected payment of $c_2/(1-\delta^2)$, so that the expected payoff to the principal is $R_2 - c_2/(1-\delta^2) = (1/\delta - 1 + \delta) - (1/\delta - 2 + \delta)/(1+\delta)$. (2) The optimal separable contract can either incentivize action 1 by paying nothing or it can incentivize action 2 by setting $p_1 = 2c_2/(1-\delta)$ and $p_2 = 0$. Since
	\begin{align*}
	R_2 - q_{2,1} p_1 = (\frac{1}{\delta}-1+\delta) - \frac{1}{2} \frac{2c_2}{(1-\delta)} = 1 
	\end{align*}
	the expected payoff to the principal in both cases is $1$.
	
	Using (1) and (2) and setting $\delta = \frac{1}{2}(3 - \epsilon - \sqrt{\epsilon^2-10\epsilon+9})$ we have
	\begin{align*}
	\frac{OPT}{ALG} = (\frac{1}{\delta} - 1 + \delta) - \frac{\frac{1}{\delta} - 2 + \delta}{1+\delta} = 2- \epsilon.
	\end{align*}
	
	It remains to show (1) and (2). For (1) denote the payments in the optimal contract for outcomes (1,0), (0,1), and (1,1) by $p_1, p_2, p_{1,2}$. The optimal contract can incentivize action $2$ via $p_1 > 0$ and $p_2 = p_{1,2} = 0$ as long as
	\begin{align*}
	&q_{2,1}(1-q_{2,2}) p_1 - c_2 \geq q_{1,1}(1-q_{1,2}) p_1\\
	\Leftrightarrow \quad &p_1 \geq \frac{c_2}{q_{2,1}(1-q_{2,2}) - q_{1,1}(1-q_{1,2})} = \frac{4 c_2}{1-\delta^2}
	\end{align*}
	Setting $p_1 =  4c_2/(1-\delta^2)$ leads to an expected payment of $q_{2,1}(1-q_{2,2}) p_1 = c_2/(1-\delta^2)$.
	
	For (2) denote the payments of the optimal separable contract by $p_1$ and $p_2$ and note that the optimal separable contract either has $p_1 > 0$ and $p_2 = 0$ or it has $p_1 = 0$ and $p_2 > 0$. In the former case the incentive constraint is 
	\begin{align*}
	q_{2,1} p_1 - c_2 \geq q_{1,1} p_1
	\end{align*}
	and in the latter it is
	\begin{align*}
	q_{2,2} p_2 - c_2 \geq q_{1,2} p_2.
	\end{align*}
	
	Note that since $q_{1,2} = 1-\delta/2 > 1/2 = q_{1,2}$ it is impossible to incentivize action 2 by having only $p_2 > 0$. In the other case, where only $p_1 > 0$, the smallest $p_1$ that satisfies the incentive constraint is $p_1 = c_2/(q_{2,1}-q_{1,1}) = 2c_2/(1-\delta)$. 
\end{proof}

\section{Proofs of auxiliary lemmas in Section~\ref{sec:black-box}}
\label{appx:black-box}
In this appendix we provide proofs for Lemma~\ref{lem:sample-bound}, Lemma~\ref{lem:ic-is-preserved-approximately}, Lemma~\ref{lem:delta-ic-is-preserved}, and Lemma~\ref{lem:noisy-2}.

\begin{proof}[Proof of Lemma~\ref{lem:sample-bound}]
Note that with $s = (3 \log(\frac{2n}{\eta\gamma}))/(\eta \epsilon^2)$ we have $\gamma = \frac{n}{\eta} \cdot 2 \exp(-\eta s \epsilon^2/3)$. Further note that since $q_{i,S} \geq \eta$ for all $i \in [n]$ and $S \subseteq M$ each action can assign positive probability to at most $1/\eta$ sets $S$. Finally, for all $i \in [n], S \subseteq M$ such that $q_{i,S} = 0$ we have $\tilde{q}_{i,S} = 0$. So, by the union bound, it suffices to show that for each of the at most $n/\eta$ pairs $i, S$ with $q_{i,S} > 0$ the probability with which $\tilde{q}_{i,S}$ does not fall into $[(1-\epsilon) q_{i,S},(1+\epsilon) q_{i,S}]$ is at most $2 \exp(-\eta s \epsilon^2/3)$.

Consider any such pair $i,S$. Let $X_{i,S}$ denote the random variable that counts the number of times set $S$ was returned in the $s$ queries to action $i$. Then $\tilde{q}_{i,S} = X_{i,S}/s$ and $\mathbb{E}[X] = s q_{i,S}$. So, using Chernoff's bound,
\begin{align*}
\Pr[\tilde{q}_{i,S} \not\in [(1-\epsilon)q_{i,S},(1+\epsilon)q_{i,S}]] &= \Pr[|X_{i,S} - \mathbb{E}[X_{i,S}]| \geq \epsilon] \\
&\leq 2 \exp(- \eta s \epsilon^2/3),
\end{align*}
as claimed.
\end{proof}

\begin{proof}[Proof of Lemma~\ref{lem:ic-is-preserved-approximately}]
Let $a_i$ be the action that is incentivized by $p$ under the actual probabilities $Q$, and let $a_{i'}$ be any other action. Then,
\begin{align*}
\sum_{S \subseteq M} \tilde{q}_{i,S} p_{i,S} - c_i + 2 \epsilon 
&\geq (1-\epsilon) \sum_{S \subseteq M} q_{i,S} p_{i,S} - c_i + 2 \epsilon\\
&\geq \sum_{S \subseteq M} q_{i,S} p_{i,S} - c_i  + \epsilon\\
&\geq \sum_{S \subseteq M} q_{i',S} p_{i',S} - c_{i'} + \epsilon\\
&\geq (1+\epsilon) \sum_{S \subseteq M} q_{i',S} p_{i',S} - c_{i'} \\
&\geq \sum_{S \subseteq M} \tilde{q}_{i',S} p_{i',S} - c_{i'},
\end{align*}
where we used the bounds on the probabilities in the first and last step, that we are considering normalized settings in the second and fourth step, and the IC constraint in the third step.
\end{proof}

\begin{proof}[Proof of Lemma~\ref{lem:delta-ic-is-preserved}]
Let $a_i$ be the action that is incentivized by $\tilde{p}$ under the empirical probabilities $\tilde{Q}$, and let $a_{i'}$ be any other action. Then,
\begin{align*}
\sum_{S \subseteq M} q_{i,S} p_{i,S} - c_i + \delta + 2 \epsilon 
&\geq (1+\epsilon) \sum_{S \subseteq M} q_{i,S} p_{i,S} - c_i + \delta + \epsilon \\
&\geq \sum_{S \subseteq M} \tilde{q}_{i,S} p_{i,S} - c_i  + \delta + \epsilon\\
&\geq \sum_{S \subseteq M} \tilde{q}_{i',S} p_{i',S} - c_{i'} + \epsilon\\
&\geq (1-\epsilon) \sum_{S \subseteq M} q_{i',S} p_{i',S} - c_{i'} + \epsilon \\
&\geq \sum_{S \subseteq M} q_{i',S} p_{i',S} - c_{i'},
\end{align*}
where we used that we are considering normalized settings in the first and the last step, the bounds on the probabilities in the second and fourth step, and the $\delta$-IC constraint in the third step.
\end{proof}

\begin{proof}[Proof of Lemma~\ref{lem:noisy-1}]
We have,
\begin{align*}
\tilde{\Pi} &= \sum_{S \subseteq M} \tilde{q}_{i,S} r_S - \sum_{S \subseteq M} \tilde{q}_{i,S} p_{i,S}\\
&\leq (1+\epsilon) \sum_{S \subseteq M} q_{i,S} r_S - (1-\epsilon) \sum_{S \subseteq M} q_{i,S} p_{i,S}\\
&\leq \sum_{S \subseteq M} q_{i,S} r_S -  \sum_{S \subseteq M} q_{i,S} p_{i,S} + 2\epsilon\\
&= \Pi + 2 \epsilon,
\end{align*}
where we used the bounds on the payments in the first step and that we are considering normalized settings in the second.
\end{proof}

\begin{proof}[Proof of Lemma~\ref{lem:noisy-2}]
We have,
\begin{align*}
P &= \sum_{S \subseteq M} q_{i,S} r_S - \sum_{S \subseteq M} q_{i,S} p_{i,S}\\
&\leq \frac{1}{1-\epsilon} \sum_{S \subseteq M} \tilde{q}_{i,S} r_S - \frac{1}{1+\epsilon} \sum_{S \subseteq M} \tilde{q}_{i,S} p_{i,S}\\
&\leq (1+2\epsilon) \sum_{S \subseteq M} \tilde{q}_{i,S} r_S -  (1-\epsilon) \sum_{S \subseteq M} q_{i,S} p_{i,S}  \\
&= \Pi + 3 \epsilon,
\end{align*}
where we used the bounds on the probability in the first step, and $1/(1-\epsilon) \leq 1+2\epsilon$ as well as $1/(1+\epsilon) \geq 1-\epsilon$ for all $\epsilon \leq 1/2$. 
\end{proof}

\end{document}